\documentclass[11pt]{article}
\usepackage{fullpage}
\pdfoutput=1
\usepackage{graphicx}
\usepackage{amsmath}
\usepackage{amssymb}
\usepackage{subfigure}
\usepackage{epstopdf}
\newtheorem{theorem}{Theorem}
\newtheorem{corollary}{Corollary}
\newenvironment{proof}{{\noindent \em Proof.}}{\hfill$\Box$\\}

\newcommand{\Oh}[1]
    {\ensuremath{\mathcal{O} \hspace{-.5ex} \left( {#1} \right)}}

\newcommand{\occ}[2]
    {\ensuremath{\mathrm{occ} \hspace{-.5ex} \left( {#1}, {#2} \right)}}
\newcommand{\socc}[2]
    {\ensuremath{\mathrm{occ} \left( {#1}, {#2} \right)}}

\newcommand{\lmax}{\ell_{\max}}
\newcommand{\pmin}{p_{\min}}
\newcommand{\clong}{c_{\mathrm{long}}}
\newcommand{\Codes}{\mathsf{Codes}}
\newcommand{\Symb}{\mathsf{Symb}}
\newcommand{\sR}{\mathsf{sR}}
\newcommand{\first}{\mathsf{first}}
\newcommand{\hash}{\mathsf{hash}}
\newcommand{\ihash}{\mathsf{ihash}}
\newcommand{\HH}{\mathcal{H}}
\newcommand{\LL}{\mathcal{L}}

\newcommand{\EsWiki}{\texttt{EsWiki}}
\newcommand{\EsInv}{\texttt{EsInv}}
\newcommand{\Indo}{\texttt{Indo}}

\begin{document}

\title{Efficient and Compact Representations of Prefix Codes
\thanks{Funded in part by Millennium Nucleus Information and
Coordination in Networks ICM/FIC RC130003 for Gonzalo Navarro, and by MINECO (PGE and FEDER) [TIN2013-46238-C4-3-R, TIN2013-47090-C3-3-P]; 
 CDTI, AGI, MINECO [CDTI-00064563/ITC-20133062]; ICT COST Action IC1302; Xunta de Galicia (co-founded with FEDER) [GRC2013/053], and AP2010-6038 (FPU Program) for Alberto Ord\'o\~nez.}
~\thanks{Early partial versions of this work appeared in {\em Proc. SOFSEM
2010} \cite{GNN10} and {\em Proc. DCC 2013} \cite{NO13}.}
~\thanks{Corresponding author: Alberto Ord\'o\~nez, email: alberto.ordonez@udc.es.}
}

\author{Travis Gagie$^\star$   ~~~~~
        Gonzalo Navarro$^\dag$ ~~~~~
    	Yakov Nekrich$^\ddag$  ~~~~~
	Alberto Ord\'o\~nez$^+$ \\  \ \\
{\small $^\star$ Department of Computer Science, University of Helsinki,
Finland} \\
{\small $^\dag$ Department of Computer Science, University of Chile, Chile} \\
{\small $^\ddag$ David R. Cheriton School of Computer Science, University of Waterloo,
Canada} \\
{\small $^+$ Database Laboratory, Universidade da Coru\~na, Spain}}

\date{}

\maketitle

\begin{abstract}
Most of the attention in statistical compression is given to the space used
by the compressed sequence, a problem completely solved with optimal prefix
codes. However, in many applications, the storage space used to represent the 
prefix code itself can be an issue. In this paper we introduce 
and compare several techniques to store prefix codes. Let $N$ be the sequence
length and $n$ be the alphabet size. Then a naive storage of an optimal prefix 
code uses $\Oh{n\log n}$ bits. Our first technique shows how to use 
$\Oh{n\log\log(N/n)}$ bits to store the optimal prefix code. Then we introduce an
approximate technique that, for any $0<\epsilon<1/2$, takes
$\Oh{n \log \log (1 / \epsilon)}$ bits to store a prefix code with average 
codeword length within an additive $\epsilon$ of the minimum. Finally, a
second approximation takes, for any constant \(c > 1\), $\Oh{n^{1 / c} \log n}$
bits to store a prefix code with average codeword length at most $c$ times 
the minimum. In all cases, our data structures allow encoding and decoding of
any symbol in $\Oh{1}$ time. We experimentally compare our new techniques with
the state of the art, showing that we achieve 6--8-fold space reductions, 
at the price of a slower encoding (2.5--8 times slower) and decoding (12--24 
times slower). The approximations further reduce this space and improve the
time significantly, up to recovering the speed of classical implementations, 
for a moderate penalty in the average code length. 
As a byproduct, we compare various heuristic, approximate, and
optimal algorithms to generate length-restricted codes, showing that the
optimal ones are clearly superior and practical enough to be implemented.
\end{abstract}

\section{Introduction} \label{sec:intro}

Statistical compression is a well-established branch of Information Theory.
Given a text $T$ of length $N$, over an alphabet of $n$ symbols 
$\Sigma= \{ a_1, \ldots, a_n \}$ with relative frequencies
$P = \langle p_1, \ldots, p_n \rangle$ in $T$ (where $\sum_{i=1}^n p_i=1$), 
the binary {\em empirical 
entropy} of the text is $\HH(P) = \sum_{i=1}^n p_i \lg (1/p_i)$, where $\lg$
denotes the logarithm in base 2. An {\em instantaneous code} assigns a binary
code $c_i$ to each symbol $a_i$ so that the symbol can be decoded as soon as
the last bit of $c_i$ is read from the compressed stream. An optimal (or
minimum-redundancy) instantaneous code (also called a prefix code) like 
Huffman's \cite{Huf52} finds a prefix-free set of codes $c_i$ of length
$\ell_i$, such that its average length $\LL(P) = \sum_{i=1}^n p_i\ell_i$ is
minimal and satisfies $\HH(P) \le \LL(P) < \HH(P)+1$. This guarantees that
the encoded text uses less than $N(\HH(P)+1)$ bits. Arithmetic codes achieve
less space, $N\HH(P)+2$ bits, however they are not instantaneous, which
complicates and slows down both encoding and decoding.

In this paper we are interested in instantaneous codes. In terms of the
{\em redundancy} of the code, $\LL(P)-\HH(P)$, Huffman codes are optimal and
the topic can be considered closed.
How to store 
{\em the prefix code itself}, however, is much less studied. It is not hard to
store it using $\Oh{n\log n}$ bits, and this is sufficient when $n$ is much
smaller than $N$.
There are several scenarios, however, where the storage of the code itself is
problematic. One example is word-based compression, which is a standard to
compress natural language text \cite{BSTW86,Mof89}. Word-based Huffman
compression not only performs very competitively, offering compression ratios
around 25\%, but also benefits direct access \cite{WMB99}, text searching 
\cite{MNZBY00}, and indexing \cite{BFLN12}. In this case the alphabet
size $n$ is the number of distinct words in the text, which can reach many
millions. Other scenarios where large alphabets arise are the compression of
East Asian languages and general numeric sequences. Yet another
case arises when the text is short, for example when it is cut
into several pieces that are statistically compressed independently, for
example for compression boosting \cite{FGMS05,KP11} or for interactive
communications or adaptive compression \cite{BFNP07}. The more effectively the 
codes are stored, the finer-grained can the text be cut.

During encoding and decoding, the code must be maintained in main
memory to achieve reasonable efficiency, whereas the plain or the compressed
text can be easily read or written in streaming mode. Therefore, the size of the
code, and not that of the text, is what poses the main memory requirements for
efficient compression and decompression. This is particularly stringent on
mobile devices, for example, where the supply of main memory is comparatively
short. With the modern trend of embedding small devices and sensors in all 
kinds of objects (e.g., the ``Internet of Things''%
\footnote{\tt http://en.wikipedia.org/wiki/Internet\_of\_Things}), those
low-memory scenarios may become common.

In this paper we obtain various relevant results of theoretical and practical
nature about how to store a code space-efficiently, while also considering the
time efficiency of compression and decompression. Our specific contributions
are the following.

\begin{enumerate}
\item In Section~\ref{sec:exact}
we show that it is possible to store an optimal prefix code within
$\Oh{n\log \lmax}$ bits, where $\lmax = \Oh{\min(n,\log N)}$ is the 
maximum length of a code (Theorem~\ref{thm:exact}). Then we refine the space
to $\Oh{n\log\log(N/n)}$ bits (Corollary~\ref{cor:exact}). 
Within this space, encoding and decoding are carried out in constant time on
a RAM machine with word size $w = \Omega(\log N)$.
The result is obtained by using canonical Huffman codes \cite{SK64}, 
fast predecessor data structures \cite{FW93,PT08} to find code lengths, and 
multiary wavelet trees 
\cite{GGV03,FMMN07,BN12} to represent the mapping between codewords and symbols.
\item In Section~\ref{sec:additive}
we show that, for any \(0 < \epsilon < 1 / 2\),
it takes $\Oh{n \log \log (1 / \epsilon)}$ bits to store a
prefix code with average codeword length at most $\LL(P)+\epsilon$. Encoding
and decoding can be carried out in constant time on
a RAM machine with word size $w = \Omega(\log n)$.
Thus, if we can tolerate a small constant additive increase in the average 
codeword length, we can store a prefix code using only $\Oh{n}$ bits. 
We obtain this result by building on the above scheme, where we use
length-limited optimal prefix codes \cite{ML01} with a carefully chosen
$\lmax$ value.
\item In Section~\ref{sec:multiplicative}
we show that, for any constant \(c > 1\), it takes $\Oh{n^{1 / c} \log n}$ bits
to store a prefix code with average codeword length at most $c\,\LL(P)$.
Encoding and decoding can be carried out in constant time on
a RAM machine with word size $w = \Omega(\log n)$.
Thus, if we can tolerate a small constant
multiplicative increase, we can store a prefix code in $o(n)$ bits.
To achieve this result, we only store the codes that are shorter than about
$\lmax/c$, and use a simple code of length $\lmax+1$ for the others.
Then all but the shortest codewords need to be explicitly represented.
\item In Section~\ref{sec:experiments}
we engineer and implement all the schemes above and compare them with careful
implementations of state-of-the-art optimal and suboptimal codes.
Our model representations are shown to use 6--8 times less space than
classical ones, at the price of being several times slower for
compression (2.5--8 times) and decompression (12--24 times). 
The additive approximations reduce these spaces up to a half and the times
by 20\%--30\%, at the expense of a small increase (5\%) in the redundancy.
The multiplicative approximations can obtain models of the same size
of the additive ones, yet increasing the redundancy to around 10\%.
In exchange, they are about as fast as the classical compression methods.
If we allow them increase the redundancy to 15\%--20\%, the multiplicative
approximations obtain model sizes
that are orders of magnitude smaller than classical representations.
\item As a byproduct, Section~\ref{sec:experiments} also compares varios
heuristic, approximation, and exact algorithms to generate length-restricted
prefix codes. The experiments show that the optimal algorithm is practical to 
implement and runs fast, while obtaining significantly better average code
lengths than the heuristics and the approximations. A very simple-to-program
approximation reaches the same optimal average code length in our experiments,
yet it runs significantly slower.
\end{enumerate}

Compared to early partial versions of this work \cite{GNN10,NO13}, this
article includes more detailed explanations, better implementations of our
exact scheme, the first implementations of the approximate schemes, 
the experimental study of the performance of algorithms that generate
length-limited codes, and stronger baselines to compare with.

\section{Related Work} \label{sec:related}

A simple pointer-based implementation of a Huffman tree takes $\Oh{n \log n}$ 
bits, and it is not difficult to show this is an optimal upper bound for 
storing a prefix code with minimum average codeword length.  For example, 
suppose we are given a permutation $\pi$ over $n$ symbols.  Let $P$ be the 
probability distribution that assigns probability \(p_{\pi(i)} = 1 / 2^i\) for 
\(1 \leq i < n\), and probability \(p_{\pi(n)}=1 / 2^{n - 1}\). Since $P$ is
dyadic, every optimal prefix code assigns 
codewords of length $\ell_{\pi(i)} = i$, for \(1 \leq i < n\), and
\(\ell_{\pi(n)} = n - 1\).  Therefore, given any optimal prefix code and a bit
indicating whether \(\pi (n - 1) < \pi (n)\), we can reconstruct $\pi$. Since 
there are $n!$ choices for $\pi$, in the worst case it takes \(\Omega (\log
n!) = \Omega (n \log n)\) bits to store an optimal prefix code. 

Considering the argument above, it is natural to ask whether the same lower
bound holds for probability distributions that are not so skewed, and the
answer is no. A prefix code is {\em canonical}~\cite{SK64,MT97} if a 
shorter codeword is always lexicographically smaller than a longer codeword. 
Given any prefix code, we can always generate a canonical code with the same
code lengths. Moreover, we can reassign the codewords
such that, if a symbol is lexicographically the $j$th with a codeword of
length $\ell$, then it is assigned the $j$th consecutive codeword of length
$\ell$. It is clear that it is sufficient to store the codeword length of each
symbol to be able to reconstruct such a code, and thus the code can be
represented in $\Oh{n\log \lmax}$ bits.

There are more interesting upper bounds than $\lmax \le n$. 
Katona and Nemetz~\cite{KN76} (see also Buro \cite{Bur93})
showed that, if a symbol has relative
frequency $p$, then any Huffman code assigns it a codeword of length at most
\(\lfloor \log_\phi (1 / p)\rfloor\), where \(\phi = (1+\sqrt{5})/2 \approx 
1.618\) is the golden ratio, and thus $\lmax$ is at most 
\(\lfloor\log_\phi (1 / \pmin)\rfloor\), where 
$\pmin$ is the smallest relative frequency in $P$. Note also that, since
$\pmin \ge 1/N$, it must hold $\lmax \le \log_\phi N$, therefore the
canonical code can be stored in $\Oh{n\log\log N}$ bits.

Alternatively, one can enforce a value for $\lmax$ (which must be at
least $\lceil \lg n \rceil$) and pay a price in 
terms of average codeword length. The same bound above \cite{KN76} hints at a
way to achieve any desired $\lmax$ value: artificially increase the
frequency of the least frequent symbols until the new $\pmin$ value is 
over $\phi^{-\lmax}$, and then an optimal prefix code built on the new 
frequencies will hold the given maximum code length.
Another simple technique (see, e.g., \cite{BN09}, where it was used for
Hu-Tucker codes) is to start with an optimal prefix code, and then
spot all the highest nodes in the code tree with depth $\lmax-d$ and more
than $2^d$ leaves, for any $d$. Then the subtrees of the parents of those nodes
are made perfectly balanced. A more sophisticated technique, by Milidi\'u and
Laber \cite{ML01}, yields a performance guarantee. 
It first builds a Huffman tree $T_1$, then removes all the subtrees rooted 
at depth greater than $\lmax$, builds a complete binary tree $T_2$ of height 
$h$ whose leaves are those removed from $T_1$, finds the node $v \in T_1$ at 
depth \(\lmax - h - 1\) whose subtree $T_3$'s leaves correspond to the symbols 
with minimum total probability, and finally replaces $v$ by a new node whose 
subtrees are $T_2$ and $T_3$. They
show that the resulting average code length is at most $\LL(P) +1/ \phi^{\lmax - \lceil \lg (n + \lceil \lg n \rceil - \lmax)\rceil - 1}$.

All these approximations require $\Oh{n}$ time plus the time to build the
Huffman tree. A technique to obtain the optimal length-restricted prefix code,
by Larmore and Hirshberg \cite{LH90}, performs in $\Oh{n\,\lmax}$ time by
reducing the construction to a binary version of the coin-collector's problem. 

The above is an example of how an additive increase in the average codeword
length may yield less space to represent the code itself. Another well-known
additive approximation follows from Gilbert and Moore's proof~\cite{GM59} that 
we can build an alphabetic prefix code with average codeword length less than 
$\HH(P)+2$, and indeed no more than \(\LL(P) + 1\)~\cite{Nak91,She92}. In an 
alphabetic prefix code, the lexicographic order of the codewords is the same 
as that of the source symbols, so we need to store only the code tree and not the 
assignment of codewords to symbols. Any code tree, of $n-1$ internal nodes, 
can be encoded in $4n+o(n)$ bits so that it can be navigated in constant time 
per operation~\cite{DRR12}, and thus encoding and decoding of any symbol takes 
time proportional to its codeword length. 

Multiplicative approximations have the potential of yielding codes that can be
represented within $o(n)$ bits.
Adler and Maggs~\cite{AM01} showed it generally takes more than \((9 / 40)
n^{1 / (20 c)} \lg n\) bits to store a prefix code with average codeword
length at most \(c H (P)\). Gagie~\cite{Gag06a,Gag06b,Gag08} showed that, for 
any constant \(c \geq 1\), it takes $\Oh{n^{1 / c} \log n}$ bits to store a 
prefix code with average codeword length at most \(c H (P) + 2\).  He also 
showed his upper bound is nearly optimal because, for any positive constant 
$\epsilon$, we cannot always store a prefix code with average codeword length 
at most \(c H (P) + o (\log n)\) in $\Oh{n^{1 / c - \epsilon}}$ bits. Note
that our result does not have the additive term ``$+2$'' in addition to the
multiplicative term, which is very relevant on low-entropy texts.

\section{Representing Optimal Codes} \label{sec:exact}


\begin{figure}[t]
\begin{center}
	\subfigure[]{
		\includegraphics[width=0.45\textwidth]{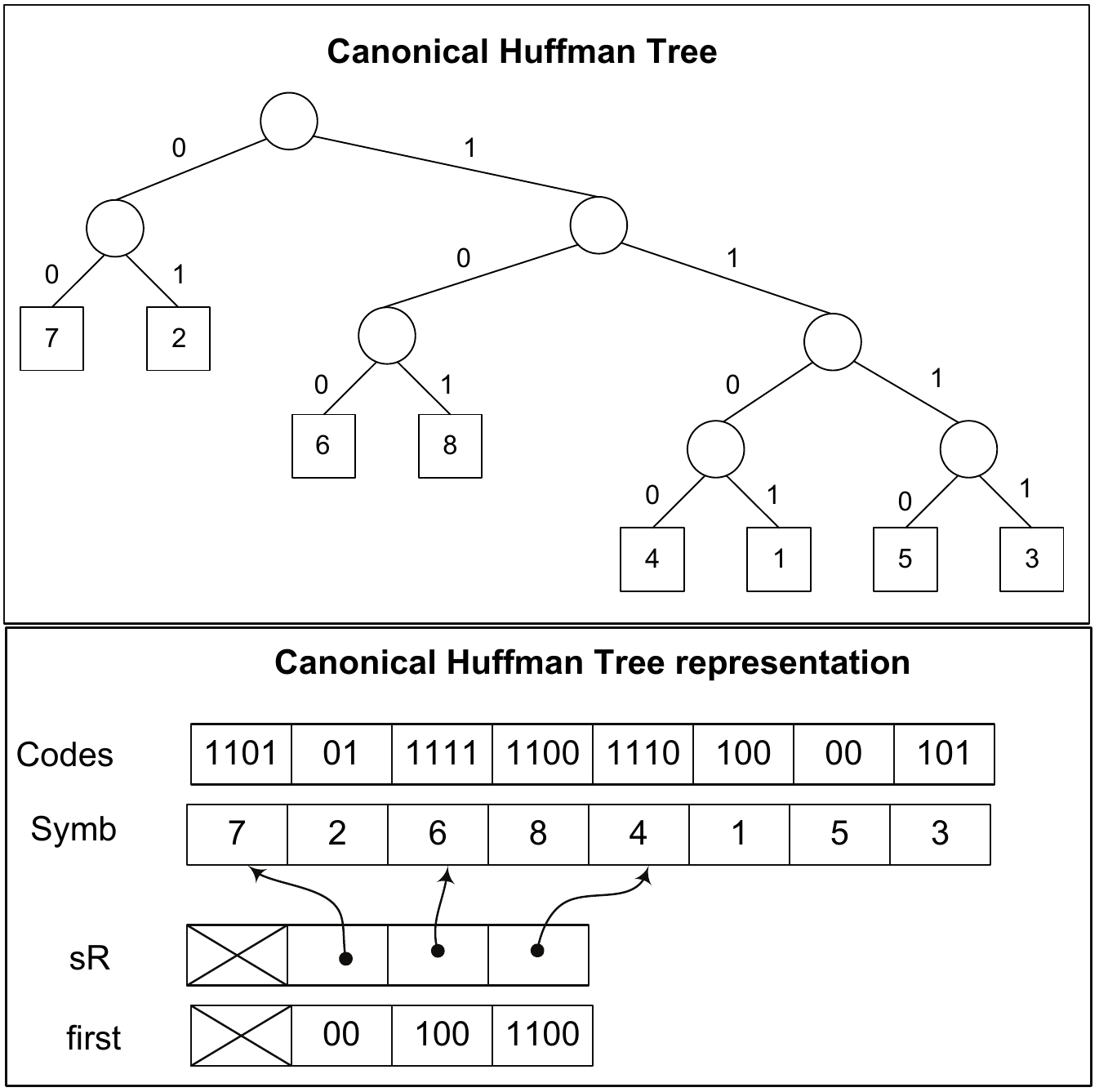}
		 \label{fig:canonical:left}
	}
	\subfigure[]{
		\vspace*{-7.5cm}
		\includegraphics[width=0.45\textwidth]{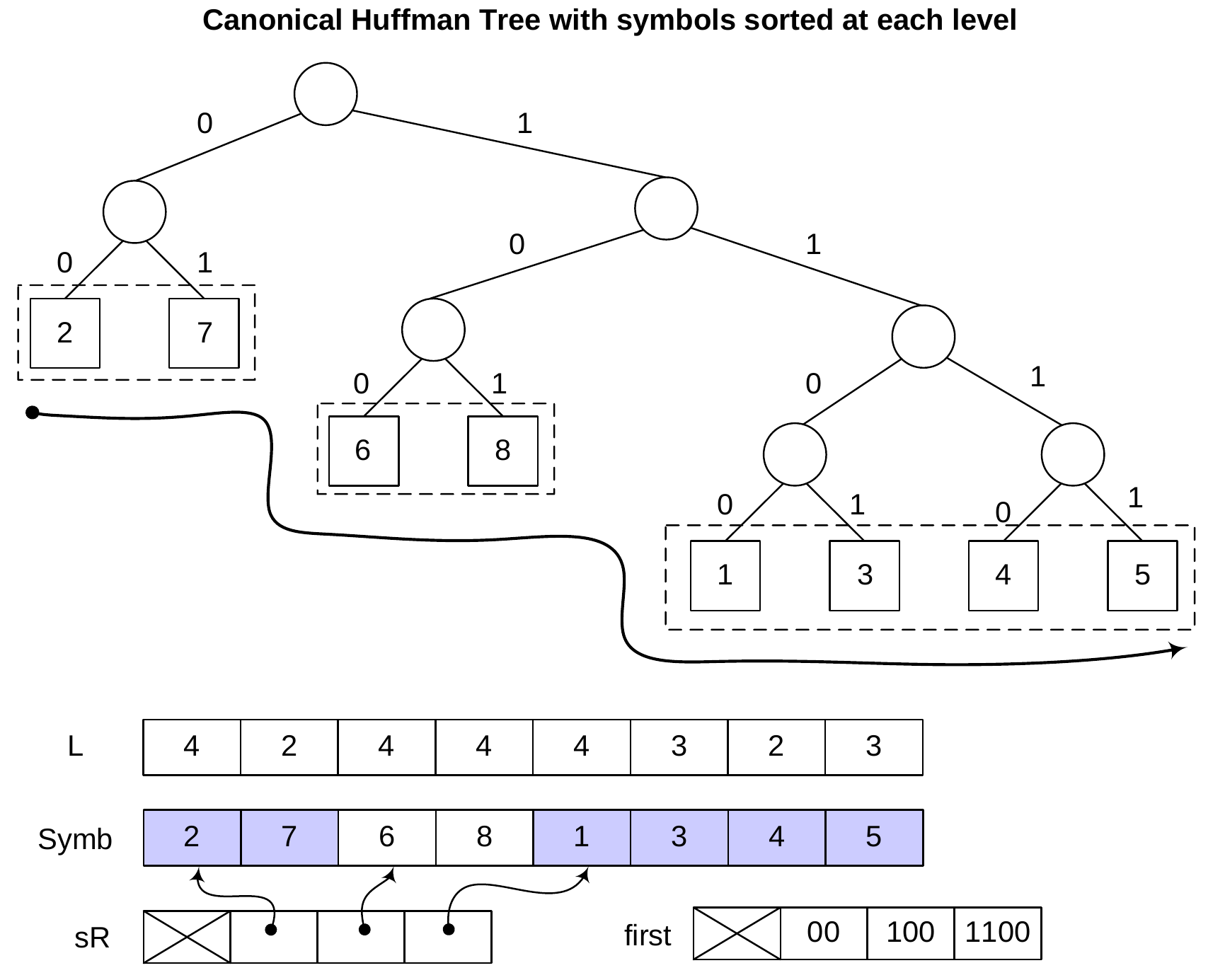}
		\label{fig:canonical:right}
	}
\end{center}
\caption{An arbitrary canonical prefix code \subref{fig:canonical:left} and the result of 
sorting the source symbols at each level \subref{fig:canonical:right}.}
\label{fig:canonical}
\end{figure}

Figure~\ref{fig:canonical:left} illustrates a canonical Huffman code.
For encoding in constant time, we can simply use an array like
$\Codes$, which stores at position $i$ the code $c_i$ of source symbol
$a_i$, using $\lmax = \Oh{\log N}$ bits for each. For decoding, the
source symbols are written in an array $\Symb$, in left-to-right order
of the leaves. This array requires $n\lg n$ bits. The access to this array is
done via two smaller arrays, which have one entry per level: $\sR[\ell]$
points to the first position of level $\ell$ in
$\Symb$, whereas $\first[\ell]$ stores the first code in level $\ell$.
The space for these two arrays is $\Oh{\lmax^2}$ bits.

Then, if we have to decode the first symbol encoded in a bitstream, we first
have to determine its length $\ell$. In our example, if the bitstream starts
with \textsf{0}, then $\ell=2$; if it starts with \textsf{10}, then $\ell=3$,
and otherwise $\ell=4$. Once the level $\ell$ is found, we read the next
$\ell$ bits of the stream in $c_i$, and decode the symbol as 
$a_i = \Symb[\sR[\ell]+c_i-\first[\ell]]$.

The problem of finding the appropriate entry in $\first$ can be recast
into a predecessor search problem \cite{GN09,KN09}. We extend all the values
$\first[\ell]$ by appending $\lmax-\ell$ bits at the end. In our
example, the values become $\mathsf{0000} = 0$, $\mathsf{1000} = 8$, and
$\mathsf{1100} = 12$. Now, we find the length $\ell$ of the next symbol by
reading the first $\lmax$ bits from the stream, interpreting it as a
binary number, and finding its predecessor value in the set. Since we have only
$\lmax=\Oh{\log N}$ numbers in the set, and each has
$\lmax=\Oh{\log N}$ bits, the predecessor search can be carried out
in constant time using fusion trees \cite{FW93} (see also Patrascu and Thorup
\cite{PT08}), within $\Oh{\lmax^2}$ bits of space.

Although the resulting structure allows constant-time encoding and decoding,
its space usage is still $\Oh{n\,\lmax}$ bits. In order to reduce it to
$\Oh{n\log\lmax}$, we will use a {\em multiary wavelet tree} data 
structure \cite{GGV03,FMMN07}. In particular, we use the version that does not
need universal tables \cite[Thm.~7]{BN12}. This structure represents a sequence
$L[1,n]$ over alphabet $[1,\lmax]$ using $n\lg\lmax +
o(n\lg\lmax)$ bits, and carries out the operations in time
$\Oh{\log\lmax/\log w}$. In our case, where $\lmax = \Oh{w}$,
the space is  $n\lg\lmax + o(n)$ bits and the time is $\Oh{1}$. The
operations supported by wavelet trees are the following: 
(1) Given $i$, retrieve $L[i]$; (2) given $i$ and $\ell \in [1,\lmax]$, 
compute $rank_\ell(L,i)$, the number of occurrences of $\ell$ in $L[1,i]$; (3) 
given $j$ and $\ell \in [1,\lmax]$, compute $select_\ell(S,j)$, the 
position in $L$ of the $j$-th occurrence of $\ell$. 

Assume that the symbols of the canonical Huffman tree are in increasing
order within each depth, as in Figure~\ref{fig:canonical:right}.%
\footnote{In fact, most previous descriptions of canonical Huffman codes assume
this increasing order, but we want to emphasize that this is essential for
our construction.}
Now, the key property is that $\Codes[i] = \first[\ell] + 
rank_\ell(L,i)-1$, where $\ell=L[i]$, which finds the code $c_i=\Codes[i]$ 
of $a_i$ in constant time. The inverse property is useful for
decoding code $c_i$ of length $\ell$: the symbol is 
$a_i = \Symb[\sR[\ell]+c_i-\first[\ell]] =
select_\ell(L,c_i-\first[\ell]+1)$. Therefore, arrays
$\Codes$, $\Symb$, and $\sR$ are not required; we can encode
and decode in constant time using just the wavelet tree of $L$ and $\first$,
plus its predecessor structure. This completes the result.

\begin{theorem} \label{thm:exact}
Let $P$ be the frequency distribution over $n$ symbols for a text of length
$N$, so that an optimal prefix code has maximum codeword length $\lmax$. Then, 
under the RAM model with computer word size $w=\Omega(\lmax)$, we can store an 
optimal prefix code using $n\lg\lmax + o(n) + \Oh{\lmax^2}$ bits, note that
$\lmax \le \log_\phi N$. Within this space, encoding and decoding any symbol 
takes $\Oh{1}$ time.
\end{theorem}

Therefore, under mild assumptions, we can store an optimal code in 
$\Oh{n\log\log N}$ bits, with constant-time encoding and decoding operations.
In the next section we refine this result further. On the other hand, note
that Theorem~\ref{thm:exact} is also valid for nonoptimal prefix codes, as
long as they are canonical and their $\lmax$ is $\Oh{w}$.

We must warn the practice-oriented reader that Theorem~\ref{thm:exact} (as well
as those to come) must be understood as a theoretical result. 
As we will explain in Section~\ref{sec:experiments}, other structures with
worse theoretical guarantees perform better in practice than those chosen to
obtain the best theoretical results. Our engineered implementation of 
Theorem~\ref{thm:exact} reaches $\Oh{\log\log N}$, and even $\Oh{\log N}$,
decoding time. It does, indeed, use much less space than previous model 
representations, but it is also much slower.

\section{Additive Approximation} \label{sec:additive}

In this section we exchange a small additive penalty over the optimal prefix 
code for an even more space-efficient representation of the code, while
retaining constant-time encoding and decoding.

It follows from Milidi\'u and Laber's bound \cite{ML01} that, for any
$\epsilon$ with \(0 < \epsilon < 1 / 2\), there is always a prefix code with
maximum codeword length \(\lmax = \lceil \lg n\rceil + \lceil \log_\phi (1 /
\epsilon)\rceil + 1\) and average codeword length within an additive
\[
\frac{1}{\phi^{\lmax - \lceil \lg (n + \lceil \lg n \rceil - \lmax) \rceil - 1}}
~~\le~~ \frac{1}{\phi^{\lmax-\lceil\lg n\rceil-1}}
~~\leq~~ \frac{1}{\phi^{\log_\phi (1 / \epsilon)}}
~~=~~ \epsilon
\]
of the minimum $\LL(P)$. The techniques described in Section~\ref{sec:exact}
give a way to store such a code in $n\lg\lmax+\Oh{n+\lmax^2}$ bits, with
constant-time encoding and decoding. In order to reduce the space, we note
that our wavelet tree representation \cite[Thm.~7]{BN12} in fact uses 
$n\HH_0(L) + o(n)$ bits when $\lmax = \Oh{w}$. Here $\HH_0(L)$ denotes
the empirical zero-order entropy of $L$. Then we obtain the following result.

\begin{theorem} \label{thm:additive}
Let $\LL(P)$ be the optimal average codeword length for a distribution $P$
over $n$ symbols. Then, for any $0<\epsilon<1/2$, under the RAM model with 
computer word size $w=\Omega(\log n)$, we can store a prefix code over $P$ 
with average codeword length at most $\LL(P)+\epsilon$,
using $n\lg\lg(1 / \epsilon) + \Oh{n}$ 
bits, such that encoding and decoding any symbol takes $\Oh{1}$ time.
\end{theorem}

\begin{proof}
Our structure uses $n\HH_0(L) + o(n) + \Oh{\lmax^2}$ bits, which is
$n\HH_0(L) + o(n)$ because $\lmax = \Oh{\log n}$. To complete the proof it is
sufficient to show that 
\(\HH_0 (S) \le \lg\lg(1/\epsilon) + \Oh{1}\).

To see this, consider $L$ as two interleaved subsequences, $L_1$ and $L_2$, of
length $n_1$ and $n_2$, with $L_1$ containing those lengths \(\le \lceil \lg n
\rceil\) and $L_2$ containing those greater. Thus $n\HH_0(L) \le n_1\HH_0(L_1) +
n_2\HH_0(L_2) + n$ (from an obvious encoding of $L$ using $L_1$, $L_2$, and a
bitmap).

Let us call $\occ{\ell}{L_1}$ the number of occurrences of symbol $\ell$ in
$L_1$. Since there are at most $2^\ell$ codewords of length $\ell$, assume we
complete $L_1$ with spurious symbols so that it has exactly $2^\ell$
occurrences of symbol $\ell$. This completion cannot decrease $n_1\HH_0(L_1) =
\sum_{\ell=1}^{\lceil\lg n\rceil} \occ{\ell}{L_1} \lg
\frac{n_1}{\socc{\ell}{L_1}}$, as increasing some $\occ{\ell}{L_1}$ to
$\occ{\ell}{L_1}+1$ produces a difference of $f(n_1)-f(\occ{\ell}{L_1})\ge 0$,
where $f(x) = (x+1)\lg(x+1)-x\lg x$ is increasing. Hence we can assume $L_1$
contains exactly $2^\ell$ occurrences of symbol $1 \le \ell \le \lceil \lg n
\rceil$; straightforward calculation then shows \(n_1 \HH_0 (L_1) = \Oh{n_1}\).

On the other hand, $L_2$ contains at most \(\lmax - \lceil \lg n \rceil\)
distinct values, so \(\HH_0 (L_2) \leq \lg (\lmax - \lceil \lg n \rceil)\), unless
\(\lmax = \lceil \lg n \rceil\), in which case $L_2$ is empty and \(n_2 \HH_0 (L_2)
= 0\).  Thus $n_2 \HH_0(L_2) \le n_2 \lg (\lceil \log_\phi(1/\epsilon)\rceil+1)
= n_2 \lg\lg (1/\epsilon)+\Oh{n_2}$.
Combining both bounds, we get $\HH_0(L) \le \lg\lg(1/\epsilon) + \Oh{1}$
and the theorem holds.
\end{proof}

In other words, under mild assumptions, we can store a code using
$\Oh{n\log\log(1/\epsilon)}$ bits at the price of increasing the average
codeword length by $\epsilon$, and in addition have constant-time encoding and
decoding. For constant $\epsilon$, this means that the code uses just $\Oh{n}$
bits at the price of an arbitrarily small constant additive penalty over the
shortest possible prefix code. Figure \ref{fig:add} shows an example. 
Note that the same reasoning of this
proof, applied over the encoding of Theorem~\ref{thm:exact}, yields a refined
upper bound.

\begin{corollary} \label{cor:exact}
Let $P$ be the frequency distribution of $n$ symbols for a text of length $N$.
Then, under the RAM model with computer word size $w=\Omega(\log N)$, 
we can store an optimal prefix code for $P$ using 
$n\lg\lg(N/n) + \Oh{n+\log^2 N}$ bits, while
encoding and decoding any symbol in $\Oh{1}$ time.
\end{corollary}
\begin{proof}
Proceed as in the proof of Theorem~\ref{thm:additive}, using that
$\lmax \le \log_\phi N$ and putting inside $L_1$ the lengths up to
$\lceil \log_\phi n \rceil$. Then $n_1\HH(L_1) = \Oh{n_1}$ and
$n_2\HH(L_2) \le \lg\lg(N/n)+\Oh{n_2}$.
\end{proof}

\begin{figure}[p]
\begin{center}
\subfigure[]{
	\includegraphics[width=0.5\textwidth]{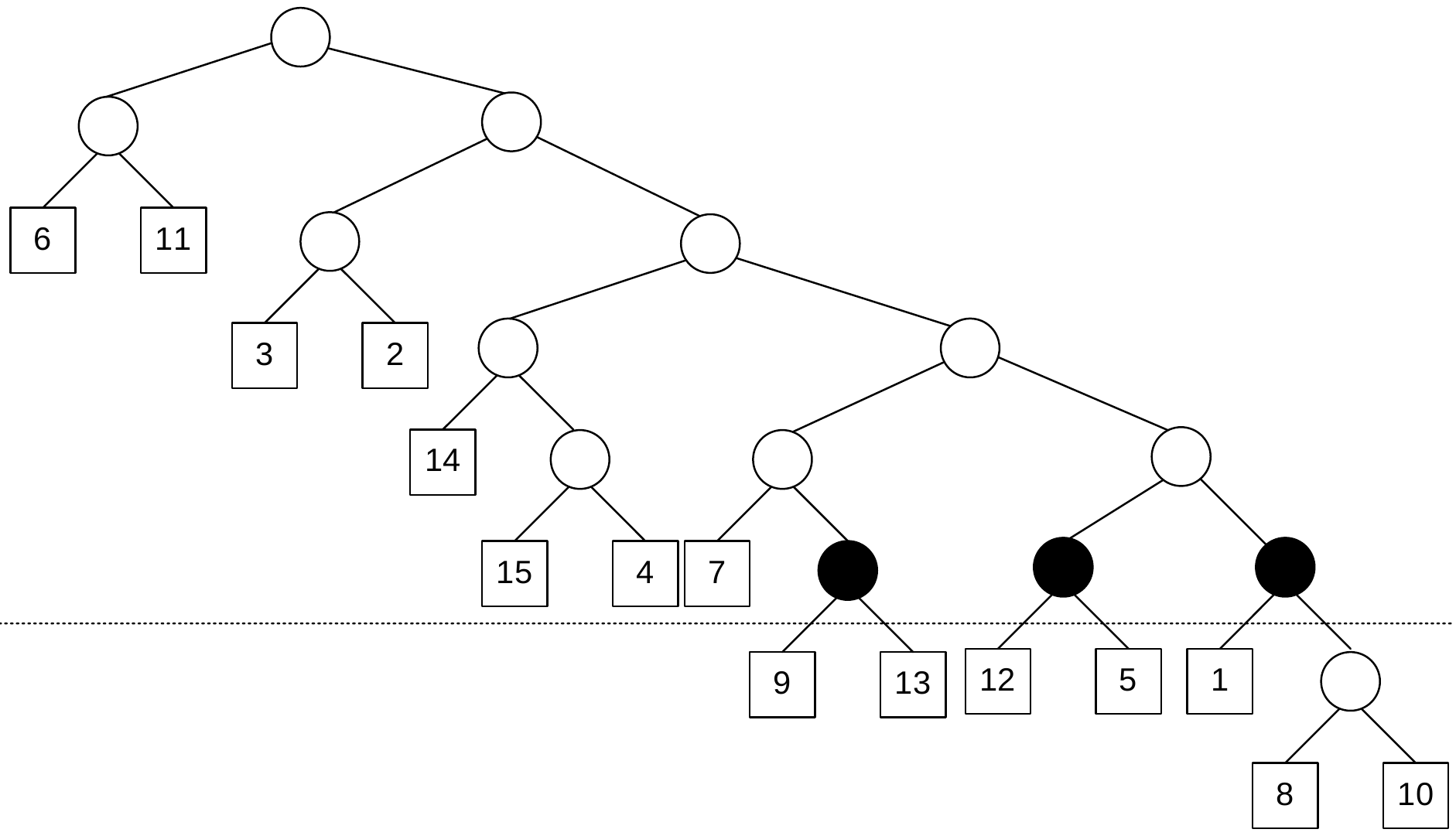}
	\label{fig:add:top}
}
\subfigure[]{
	\includegraphics[width=0.5\textwidth]{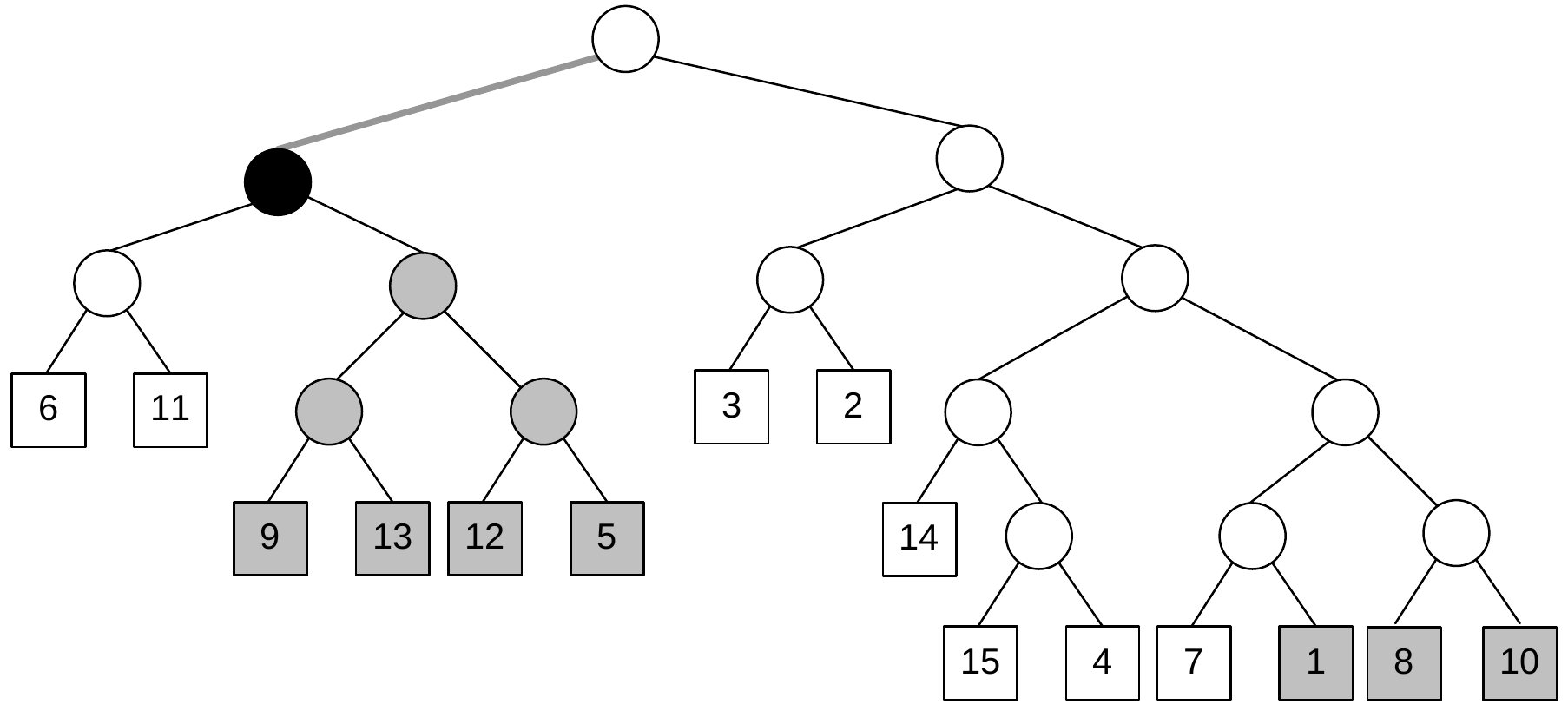}
	\label{fig:add:middle}
}
\subfigure[]{
	\includegraphics[width=0.5\textwidth]{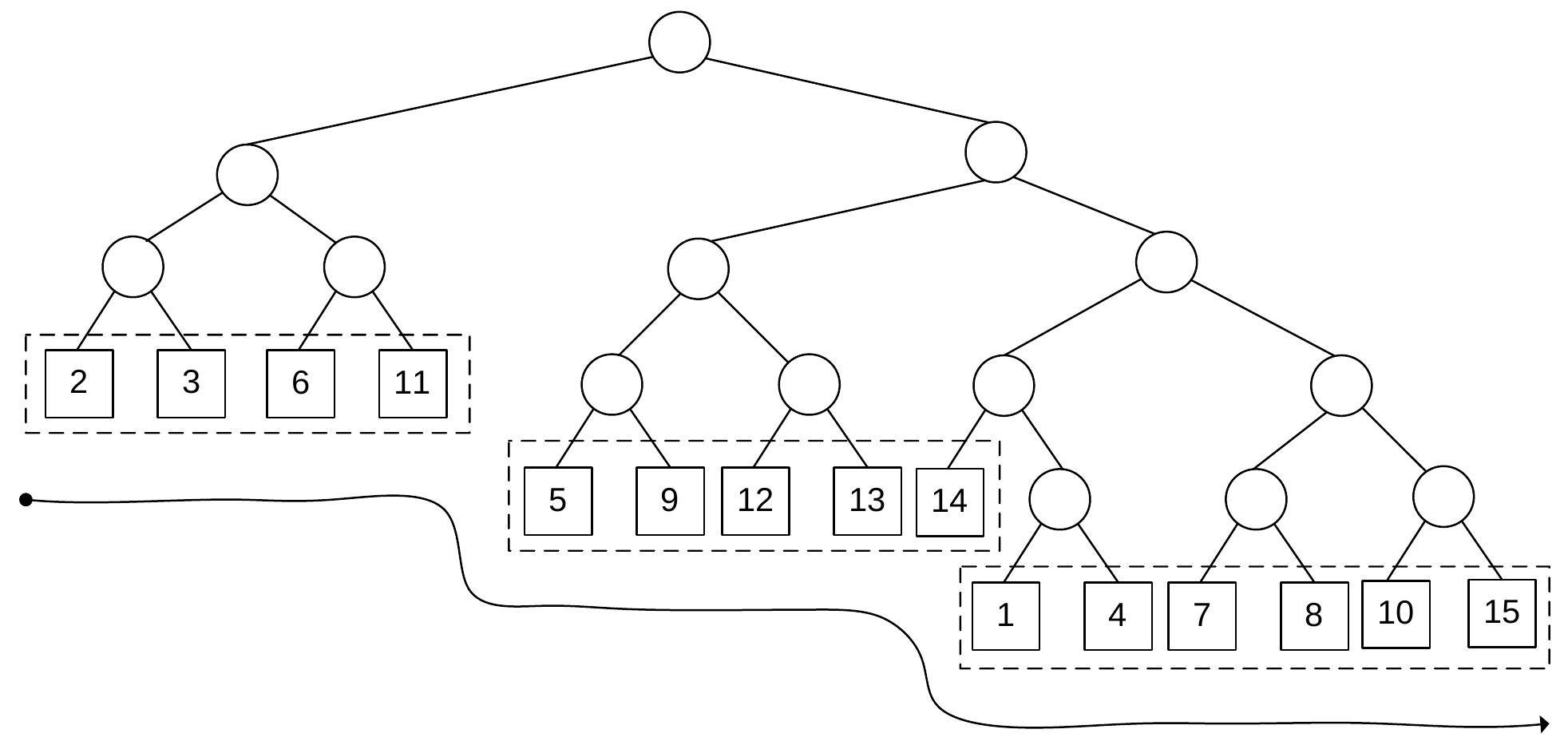}
	\label{fig:add:bottom}
}
\end{center}

\caption{An example of Milidi\'u and Laber's algorithm \cite{ML01}.
In \subref{fig:add:top}, a canonical Huffman tree. We set $l_{max}=5$ and remove all
the symbols below that level (marked with the dotted line), 
which yields three empty nodes (marked as black circles in the top tree). 
In \subref{fig:add:middle}, those black circles are replaced by the deepest symbols below
level $l_{max}$: 1, 8, and 10. The other symbols below $\lmax$,
9, 13, 12 and 5, 
form a balanced binary tree that is hung from a new node created as the left
child of the root (in black in the middle tree). 
The former left child of the root
becomes the left child of this new node. Finally, in \subref{fig:add:bottom}, we transform
the middle tree into its cannonical form, but sorting those symbols belonging to the same level in 
increasing order.}
\label{fig:add}
\end{figure}

\section{Multiplicative Approximation} \label{sec:multiplicative}

In this section we obtain a multiplicative rather than an additive
approximation to the optimal prefix code, in order to achieve a sublinear-sized 
representation of the code.
We will divide the alphabet into frequent and
infrequent symbols, and store information about only the frequent ones.

Given a constant \(c > 1\), we use Milidi\'u and Laber's algorithm \cite{ML01}
to build a prefix code with maximum codeword length 
\(\lmax = \lceil \lg n \rceil + \lceil 1 / (c - 1) \rceil + 1\)
(our final codes will have length up to $\lmax+1$).
We call a symbol's codeword {\em short} if it
has length at most \(\lmax / c + 2\), and {\em long} otherwise.  
Notice there are \(S \le 2^{\lmax / c + 2} = \Oh{n^{1 / c}}\) symbols with 
short codewords.
Also, although applying Milidi\'u and Laber's algorithm may cause some
exceptions, symbols with short codewords are usually more frequent than
symbols with long ones. We will hereafter call {\em frequent}/{\em infrequent}
symbols those encoded with short/long codewords.

Note that, if we build a canonical code, all the short codewords will precede
the long ones. We first describe how to handle the frequent symbols. A perfect 
hash data structure~\cite{FKS84} $\hash$ will map the frequent symbols in 
$[1,n]$ to the interval $[1,S]$ in constant time. The reverse mapping is done 
via a plain array $\ihash[1,S]$ that stores the original symbol that 
corresponds to each mapped symbol. We use this mapping also to 
reorder the frequent symbols so that the corresponding prefix in array $\Symb$
(recall Section~\ref{sec:exact}) reads $1,2,\ldots,S$. Thanks to this, we can
encode and decode any frequent symbol using just $\first$, $\sR$, predecessor 
structures on both of them, and the tables $\hash$ and $\ihash$. 
To encode a frequent symbol $a_i$, we find it in $\hash$, obtain the mapped 
symbol $a' \in [1,S]$, find the predecessor $\sR[\ell]$ of $a'$ and then the 
code is the $\ell$-bit integer $c_i = \first[\ell]+a'-\sR[\ell]$. To decode a 
short code $c_i$, we first find its corresponding length $\ell$ using the 
predecessor structure on $\first$, then obtain its mapped code 
$a' = \sR[\ell]+c_i-\first[\ell]$, and finally the original symbol
is $_i=\ihash[a']$. Structures $\hash$ and $\ihash$ require $\Oh{n^{1/c}\log n}$
bits, whereas $\sR$ and $\first$, together with their predecessor structures,
require less, $\Oh{\log^2 n}$ bits.

\begin{figure}[p]
\begin{center}
\subfigure[]{
	\includegraphics[width=0.54\textwidth]{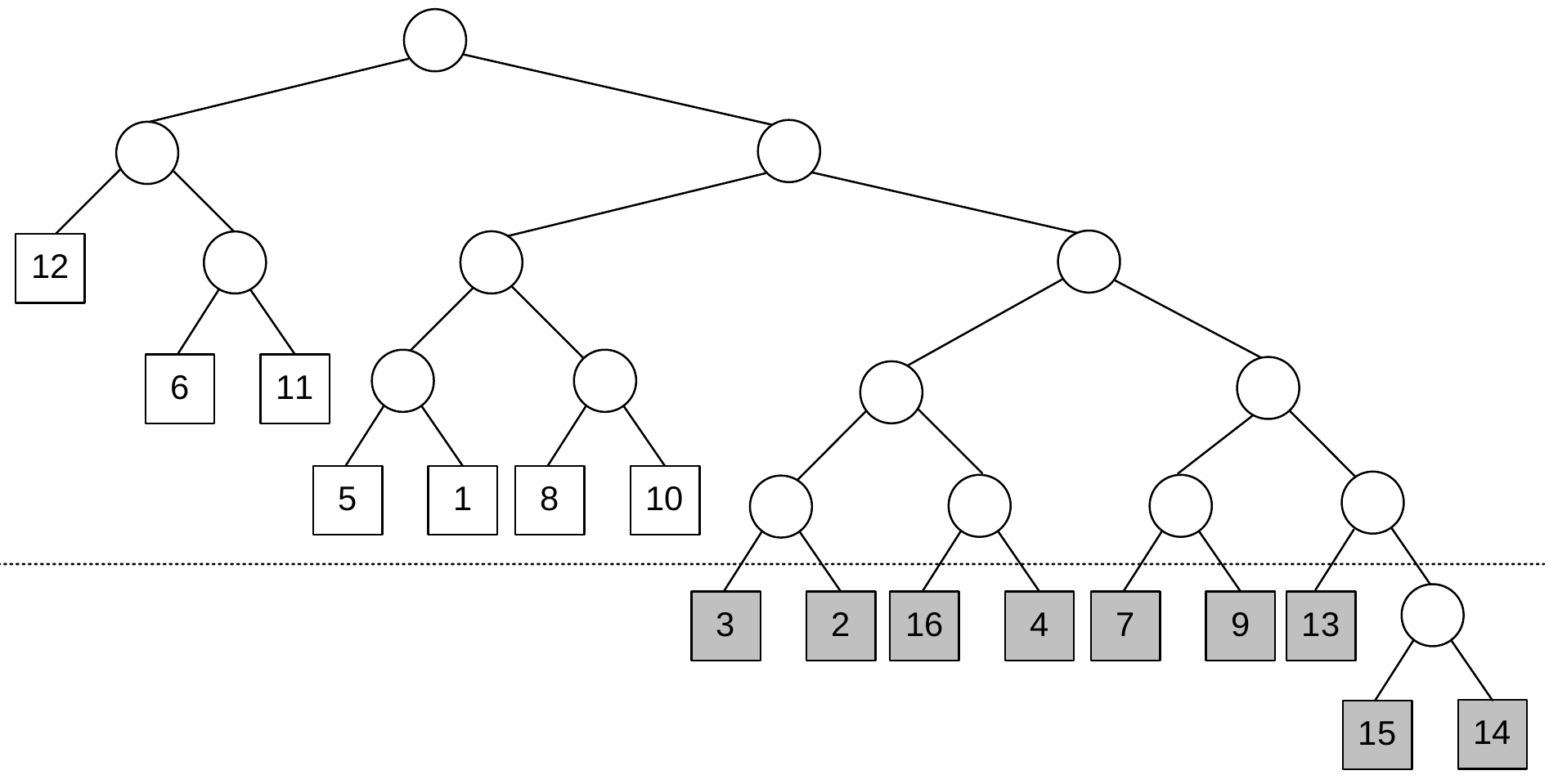}
	\label{fig:thm2:top}
}
\subfigure[]{
\includegraphics[width=0.63\textwidth]{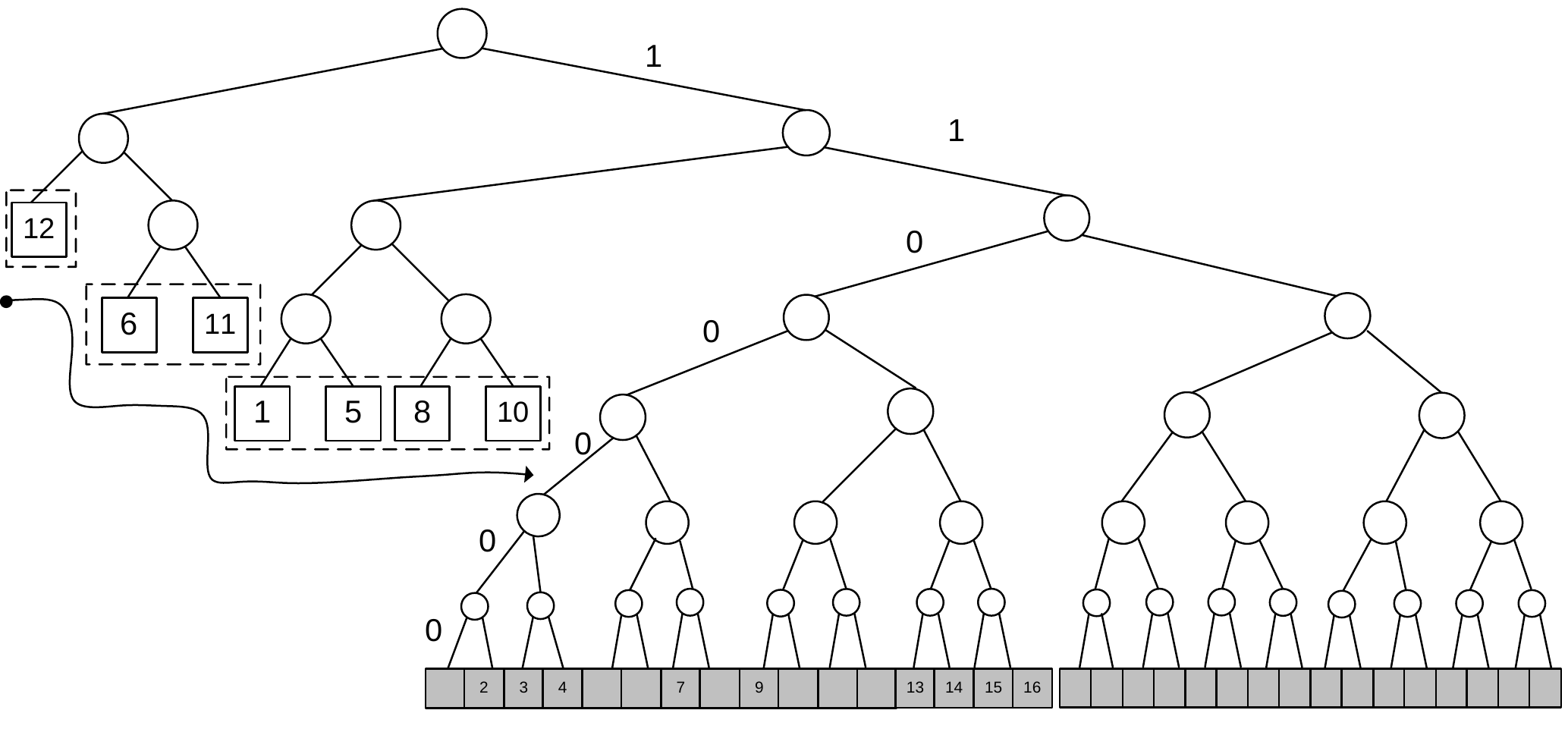}
\label{fig:thm2:middle}
}
\end{center}

\vspace*{-1.5cm}
\subfigure[]{
	\includegraphics[width=0.45\textwidth]{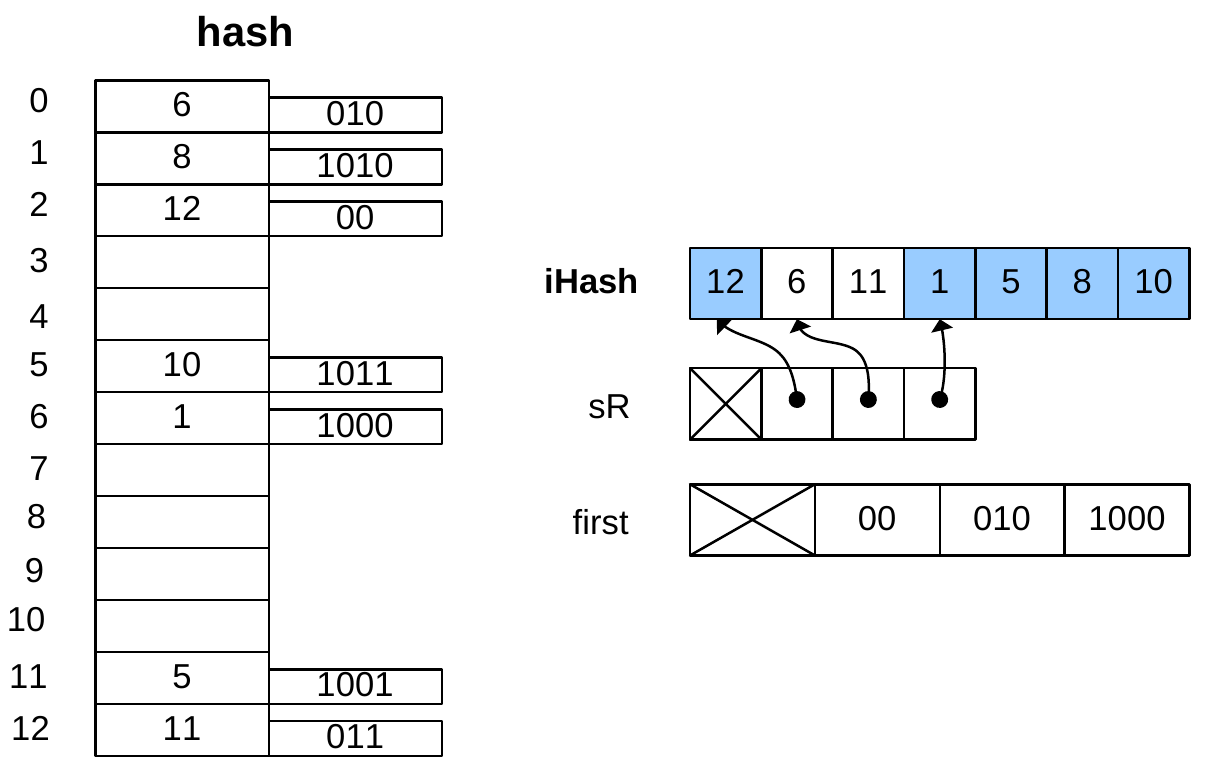}
	\label{fig:thm2:bottom}
}
\caption{An example of the multiplicative approximation, with $n=16$ and $c=3$.
The tree shown in \subref{fig:thm2:top} is the result of applying the algorithm of Milidi\'u and Laber 
to a given set of codes. Now, we set $\lmax=6$ according to our formula, 
and declare short those codewords of lengths up to $\lfloor \lmax/c \rfloor+2 = 4$. 
Short codewords (above the dashed line on top) are stored unaltered but with all symbols 
at each level sorted in increasing order \subref{fig:thm2:middle}. Long 
codewords (below the dashed line) are extended up to length $\lmax+1=7$ and 
reassigned a code according to their values in the contiguous slots of length 
7 (those in gray in the middle). Thus, given a long codeword $x$, its 
code is directly obtained as $\clong'+x-1$, where $\clong'=1100000_{2}$ is the first code of length 
$\ell_{max}+1$. In \subref{fig:thm2:bottom}, a representation of 
the hash and inverse hash to code/decode short codewords. We set the hash 
size to $m=13$ and $h(x)=(5x+7)~\textrm{mod}~m$. We store the code 
associated with each cell.}
\label{fig:thm2}
\end{figure}

The long codewords will be replaced by \emph{new} codewords, all of length
$\lmax+1$. Let $\clong$ be the first long codeword and let $\ell$ be its length. Then
we form the new codeword $\clong'$ by appending $\lmax+1-\ell$ zeros at the
end of $\clong$. The new codewords will be the $(\lmax{+}1)$-bit integers
$\clong',\clong'+1,\ldots,\clong'+n-1$. An infrequent symbol
$a_i$ will be mapped to code $\clong'+i-1$ (frequent symbols $a_i$ will leave
unused symbols $\clong'+i-1$).
Figure~\ref{fig:thm2} shows an example. 

Since \(c > 1\), we have \(n^{1 / c} < n / 2\) for sufficiently large $n$, so we can assume without loss of generality that there are fewer than $n/2$
short codewords,\footnote{If this is not the case, then $n=\Oh{1}$, so we can 
use any optimal encoding: there will be no redundancy over $\LL(P)$ and the 
asymptotic space formula for storing the code will still be valid.}
and thus there are at least $n/2$ long codewords.
Since every long codeword is replaced by at
least two new codewords, the total number of new codewords is at least $n$.
Thus there are sufficient slots to assign codewords $\clong'$ to
$\clong'+n-1$.

To encode an infrequent symbol $a_i$, we first fail to find it in table $\hash$.
Then, we assign it the $(\lmax{+}1)$-bits long codeword $\clong'+i-1$. To
decode a long codeword, we first read $\lmax+1$ bits into $c_i$. If $c_i \ge
\clong'$, then the codeword is long, and corresponds to the source symbol
$a_{c_i-\clong'+1}$. Note that we use no space to store the 
infrequent symbols. This leads to proving our result.

\begin{theorem} \label{thm:multiplicative}
Let $\LL(P)$ be the optimal average codeword length for a distribution $P$
over $n$ symbols. Then, for any constant $c>1$, under the RAM model with 
computer word size $w=\Omega(\log n)$, we can store a prefix code over $P$ 
with average codeword length at most $c\,\LL(P)$, using $\Oh{n^{1 / c} \log n}$ 
bits, such that encoding and decoding any symbol takes $\Oh{1}$ time.
\end{theorem}

\begin{proof}
Only the claimed average codeword length remains to be proved.
By analysis of the algorithm by Milidi\'u and Laber~\cite{ML01} we can see
that the codeword length of a symbol in their length-restricted code
exceeds the codeword length of the same symbol in an optimal code by at
most 1, and only when the codeword length in the optimal code is at least
\(\lmax - \lceil \log n \rceil - 1 = \lceil 1 / (c - 1) \rceil\).  Hence, the
codeword length of a frequent symbol 
exceeds the codeword length of the same symbol in an optimal code by a factor of at most \(\frac{\lceil 1 / (c - 1) \rceil + 1}{\lceil 1 / (c - 1) \rceil} \leq c\).
Every infrequent symbol is encoded with a codeword of length $\lmax+1$.
Since the codeword length of an infrequent symbol in the length-restricted
 code is more than $\lmax/c + 2$, its length in an optimal code is more than
$\lmax/c +1$.
Hence,  the codeword length of an infrequent symbol in our code is at most
 $\frac{\lmax+1}{\lmax/c+1}< c$ times greater than the codeword length of the same
symbol in an optimal code. Hence, the average codeword length for our code
is  less than $c$ times the optimal one.
\end{proof}

Again, under mild assumptions, this means that we can store a code with
average length within $c$ times the optimum, in $\Oh{n^{1/c}\log n}$ bits
and allowing constant-time encoding and decoding.

\section{Experimental Results} \label{sec:experiments}

We engineer and implement the optimal and approximate code representations
described above, obtaining complexities that are close to the theoretical
ones. We compare these with the best known alternatives to represent prefix 
codes we are aware of. Our comparisons will measure the size of the code 
representation, the encoding and decoding time and, in the case of the
approximations, the redundancy on top of $\HH(P)$.

\subsection{Implementations}
\label{sec:implem}

Our constant-time results build on two data structures. One is the multiary
wavelet tree \cite{FMMN07,BN12}. A practical study \cite{Bow10} shows that
multiary wavelet trees can be faster than binary ones, but require 
significantly more space (even with the better variants they design). To
prioritize space, we will use binary wavelet trees, which perform the 
operations in time $\Oh{\log\lmax} = \Oh{\log\log N}$.

The second constant-time data structure is the fusion tree \cite{FW93}, of
which there are no practical implementations as far as we know. Even
implementable loglogarithmic predecessor search data structures, like van Emde
Boas trees \cite{vEBKZ77}, are worse than binary search for small universes
like our range $[1,\lmax]=[1,\Oh{\log N}]$. With a simple binary search on
$\first$ we obtain a total encoding and decoding time of $\Oh{\log\log N}$,
which is sufficiently good for practical purposes. Even more, preliminary 
experiments showed that sequential search on $\first$ is about as good as 
binary search in our test collections (this is also the case with classical
representations \cite{LMSPE06}). Although sequential search costs 
$\Oh{\log N}$ time, the higher success of instruction prefetching makes it 
much faster than binary search. Thus, our experimental results use sequential search.

To achieve space close to $n\HH_0(L)$ in the wavelet tree, we use a
Huffman-shaped wavelet tree \cite{Nav14}.
The bitmaps of the wavelet tree are represented in plain form and using a space 
overhead of 37.5\% to support $rank$/$select$ operations \cite{Navjea08}. 
The total space of the wavelet tree is thus close to $1.375 \cdot n\HH_0(L)$ 
bits in practice. Besides, we enhance these bitmaps with a small additional 
index to speed up $select$ operations \cite{NPsea12.1}, which increases the
constant 1.375 to at least 1.4, or more if we want more speed. 
An earlier version 
of our work \cite{NO13} recasts this wavelet tree into a compressed permutation
representation \cite{BN13} of vector $\Symb$, which leads to a similar 
implementation.

For the additive approximation of Section~\ref{sec:additive}, we use the
same implementation as for the exact version, after modifying the code tree as
described in that section. The lower number of levels will automatically make
sequence $L$ more compressible and the wavelet tree faster.

For the multiplicative approximation of Section~\ref{sec:multiplicative}, we 
implement table $\hash$ with double hashing. The hash function is of the form 
$h(x,i) = (h_1(x)+(i-1)\cdot h_2(x))~\textrm{mod}~m$ for the $i$th trial, 
where $h_1(x)=x~\textrm{mod}~m$, 
$h_2(x)=1+(x~\textrm{mod}~(m-1))$, where $m$ is a prime number. 
 Predecessor searches over $\sR$ and $\first$ are done via binary search
since, as discussed above, theoretically better predecessor data structures
are not advantageous on this small domain.

\paragraph{Classical Huffman codes.}
As a baseline to compare with our encoding, we use the
representation of Figure~\ref{fig:canonical:left}, using $n\,\lmax$ bits
for $\Codes$, $n\lg n$ bits for $\Symb$, $\lmax^2$ bits for $\first$, and
$\lmax\lg n$ bits for $\sR$. For compression, the obvious constant-time 
solution using
$\Codes$ is the fastest one. We also implemented the fastest decompression 
strategies we are aware of, which are more sophisticated. The naive approach, 
dubbed TABLE in our experiments, consists of iteratively probing the
next $\ell$ bits from the compressed sequence, where $\ell$ is the next 
available tree depth. If the relative numeric code resulting from reading 
$\ell$ bits exceeds the number of nodes at this level, we probe the next 
level, and so on until finding the right length \cite{SK64}. 

Much research has focused on impoving upon this naive approach
\cite{MT97,MREMD03,CK85,SIEMIIPL88,HTOC95,LMSPE06}. 
For instance, one could use an additional table that takes a prefix of $b$ bits 
of the compressed sequence and tells which is the minimum code length compatible 
with that prefix. This speeds up decompression by reducing the number 
of iterations needed to find a valid code. This technique was proposed by Moffat and 
Turpin \cite{MT97} and we call it TABLE$_{S}$ in our experiments.
Alternatively, one could use a table that stores, for all the $b$-bit
prefixes, the symbols that can be directly decoded from them 
(if any) and how many bits those symbols use. Note this technique can be 
combined with TABLE$_{S}$: if no symbol can be decoded, we
use TABLE$_{S}$. In our experiments, we call TABLE$_E$ the combination of these two 
techniques.

Note that, when measuring
compression/decompression times, we will only consider the space needed for
compression/decompression (whereas our structure is a single one for both
operations).

\paragraph{Hu-Tucker codes.} 
As a representative of a suboptimal code that requires little storage space 
\cite{BNO12}, we also implement alphabetic codes, using the Hu-Tucker 
algorithm \cite{HT71,Knu73}. This algorithm takes $\Oh{n\log n}$ time and
yields the optimal alphabetic code, which guarantees an average code length
below $\HH(P)+2$. As the code is 
alphabetic, no permutation of symbols needs to be stored; the $i$th leaf of
the code tree corresponds to the $i$th source symbol. On the other hand, the
tree shape is arbitrary. We implement the code tree using succinct tree
representations, more precisely the so-called FF \cite{ACNS10}, which
efficiently supports the required navigation operations. 
This representation requires
2.37 bits per tree node, that is, $4.74n$ bits for our tree (which has $n$
leaves and $n-1$ internal nodes). FF represents general trees, so we convert
the binary code tree into a general tree using the well-known mapping
\cite{MR01}: we 
identify the left child of the code tree with the first child in the general 
tree, and the right child of the code tree with the next sibling in the 
general tree. The general tree has an extra root node whose children are the 
nodes in the rightmost path of the code tree.

With this representation, compression of symbol $c$ is carried out by starting 
from the root and descending towards the $c$th leaf. We use the number of
leaves on the left subtree to decide whether to go left or right. The
left/right decisions made in the path correspond to the code. In the general
tree, we compute the number of nodes $k$ in the subtree of the first child, 
and then the number of leaves in the code tree is $k/2$. For decompression,
we start from the root and descend left or right depending on the
bits of the code. Each time we go right, we accumulate the number of leaves on
the left, so that when we arrive at a leaf the decoded symbol is the final 
accumulated value plus 1.

\subsection{Experimental Setup}


We used an isolated AMD Phenom(tm) II X4 955 running at $800$MHz
with $8$GB of RAM memory and a ST3250318AS SATA hard disk. The operating system is
GNU/Linux, Ubuntu 10.04, with kernel 3.2.0-31-generic. All our
implementations use a single thread and are coded in {\tt C++}. The
compiler is \verb|gcc| version $4.6.3$, with \verb|-O9| optimization.
Time results refer to {\sc cpu} user time.
The stream to be compressed and decompressed is read from and written to
disk, using the buffering mechanism of the operating system. 

\begin{table}[t]
\begin{center}
\begin{tabular}{l|rrrrr}

Collection & Length & Alphabet & Entropy   & Depth     & Level entr. \\
           & ($N$)~~~& ($n$)~~~~~& ($\HH(P)$) & ($\lmax$) & ($\HH_0(L)$)~~ \\ 
\hline
\EsWiki	& 200,000,000 & 1,634,145 & 11.12 & 28 & 2.24 \\
\EsInv & 300,000,000 & 1,005,702 & 5.88  & 28 & 2.60 \\
\Indo   & 120,000,000 & 3,715,187 & 16.29 & 27 & 2.51 \\
\end{tabular}
\end{center}
\caption{Main statistics of the texts used.}
\label{tab:coll0}
\end{table}

\begin{table}[t]
\begin{center}
\begin{tabular}{l|rrrrr}
Collection & Naive~~ & Engineered & Canonical & Ours~~~ & Compressed \\
           & ($n w$)~~ & ($n\,\lmax$)~~ & ($n\lg n$)~ & ($n \HH_0(L)$) & \cite{TM00}~~~~~~\\
\hline
\EsWiki &  6.23 MB &  5.45 MB & 4.02 MB & 0.44 MB & 0.45 MB \\
\EsInv &  3.83 MB &  3.35 MB & 2.39 MB & 0.31 MB & 0.33 MB \\
\Indo   & 14.17 MB & 11.96 MB & 9.67 MB & 1.11 MB & 1.18 MB \\
\end{tabular}
\end{center}
\caption{Rough minimum size of various model representations.}
\label{tab:sizes}
\end{table}

We use three datasets\footnote{Made available in \tt http://lbd.udc.es/research/ECRPC} in our experiments. \EsWiki\ is a sequence of word
identifiers obtained by stemming the Spanish Wikipedia with the Snowball 
algorithm. Compressing natural language using word-based models is a strong
trend in text databases \cite{Mof89}.
\EsInv\ is the concatenation of differentially encoded
inverted lists of a random sample of the Spanish Wikipedia. These have
large alphabet sizes but also many repetitions, so they are highly compressible.
Finally, \Indo\ is the concatenation of the adjacency lists of Web graph 
{\tt Indochina-2004} available at {\tt http://law.di.unimi.it/datasets.php}. 
Compressing adjacency lists to zero-order entropy is a
simple and useful tool for graphs with power-law degree distributions, 
although it is usually combined with other techniques \cite{CN10}.
We use a prefix of each of the sequences to speed up experiments.

Table~\ref{tab:coll0} gives various statistics on the collections. Apart from 
$N$ and $n$, we give the empirical entropy of the sequence
($\HH(P)$, in bits per symbol or bps), the maximum length of a Huffman code 
($\lmax$), and the zero-order entropy of the sequence of levels ($\HH_0(L)$, in bps).
It can be seen that $\HH_0(L)$ is significantly smaller than
$\lg\lmax$, thus our compressed representation of $L$ can indeed be up to
an order of magnitude smaller
than the worst-case upper bound of $n\lg\lmax$ bits.

Before we compare the exact sizes of different representations, which depend on
the extra data structures used to speed up encoding and decoding,
Table~\ref{tab:sizes} gives the size of the basic data that must be stored in
each case. The first column shows $n w$, the size of a naive model 
representation using computer words of $w=32$ bits. The second shows
$n\,\lmax$, which corresponds to a more engineered representation where we use 
only the number of bits required to describe a codeword. In these two, more 
structures are needed for decoding but we ignore them. The third column gives 
$n\lg n$, which is the main space cost of a canonical Huffman tree 
representation: basically the permutation of symbols (different ones for 
encoding and decoding). The fourth column shows $n\HH_0(L)$, which is a 
lower bound on the size of our model representation (the exact value will 
depend on the desired encoding/decoding speed). These raw numbers explain why
our technique will be much more effective to represent the model than the 
typical data structures, and that we can expect up to 7--9-fold space reductions
(these will decrease to 6--8-fold on the actual structures).
Indeed, this entropy space is close to that of a sophisticated model
representation \cite{TM00} that can be used only for transmitting the model in 
compressed form; this is shown in the last column.

\subsection{Representing Optimal Codes}
\label{sec:opt-exp}

Figure~\ref{modelos.optimal} compares compression and decompression times, as a
function of the space used by the code representations, of our new
data structure (COMPR) versus the table based representations described in
Section~\ref{sec:implem} (TABLE, TABLE$_S$, and TABLE$_E$). We used sampling periods of $\{16,32,64,128\}$ for the auxiliary data
structures added to the wavelet tree bitmaps to speed up $select$ operations \cite{NPsea12.1}, and parameter $b=14$ for
table based approaches (this gave us the best time performance). 

\begin{figure}[p]

    \begin{center}

\begin{tabular}{cc}
 \subfigure{ \label{resultd}
    \includegraphics[width=0.49\textwidth]{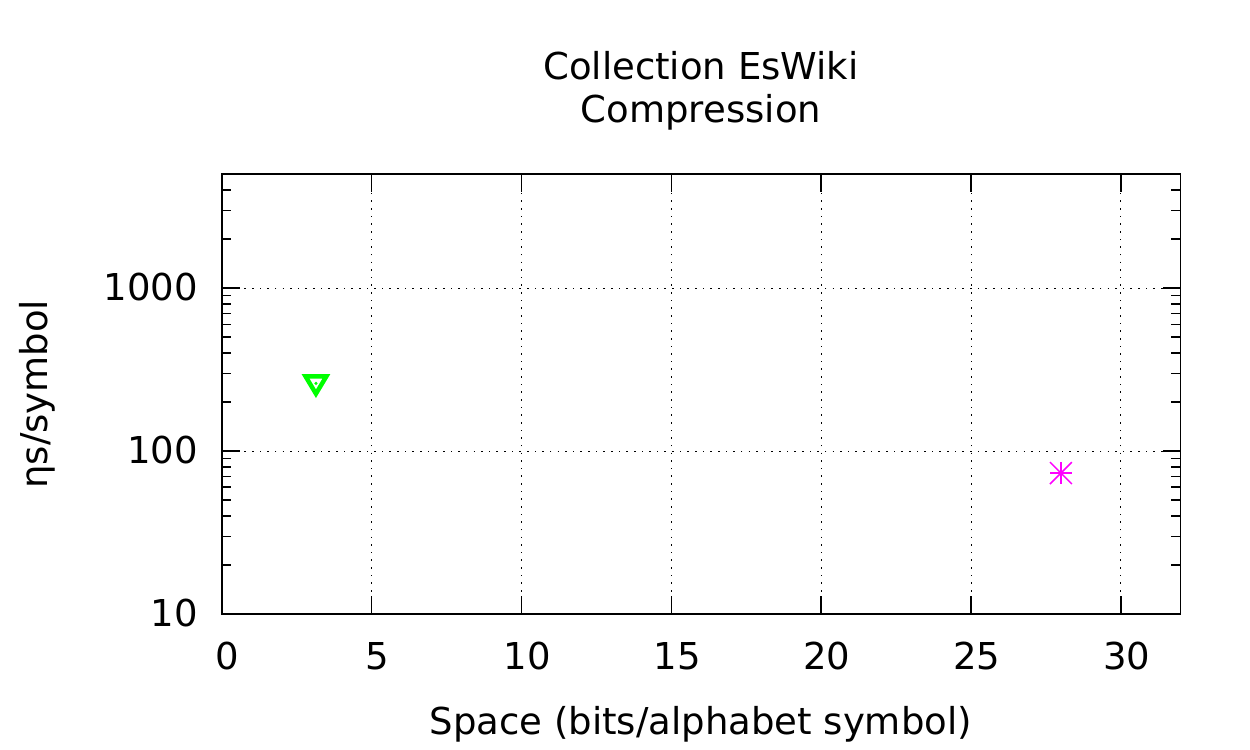}}
&
    \subfigure{ \label{resulta}
    \includegraphics[width=0.49\textwidth]{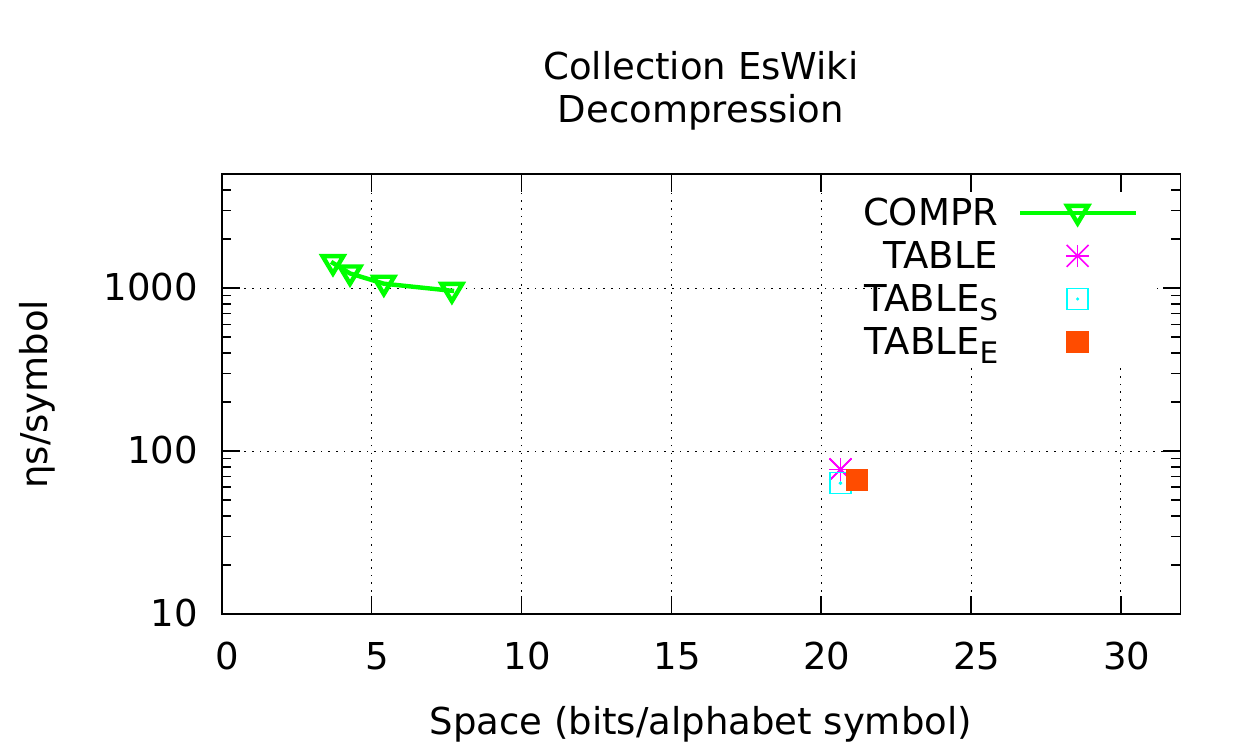}}
\end{tabular}

\begin{tabular}{cc}
    \subfigure{ \label{resulte}
    \includegraphics[width=0.49\textwidth]{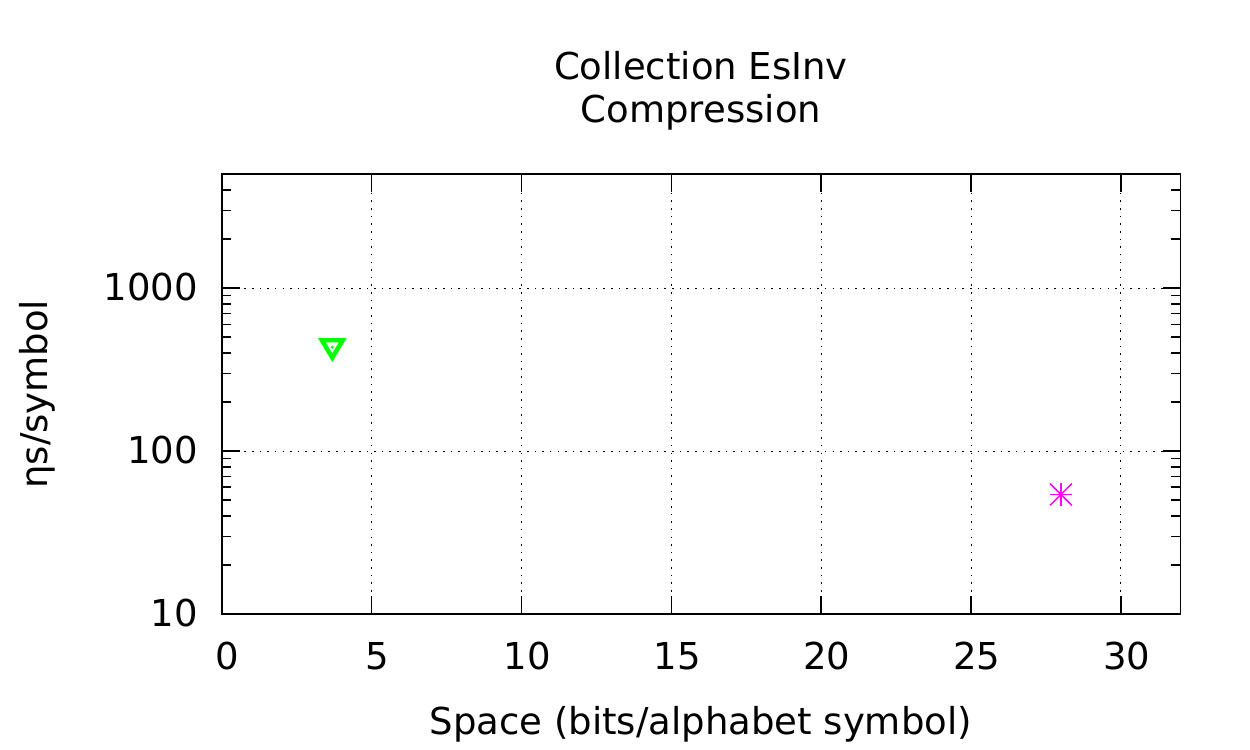}}
&
    \subfigure{ \label{resultb}
    \includegraphics[width=0.49\textwidth]{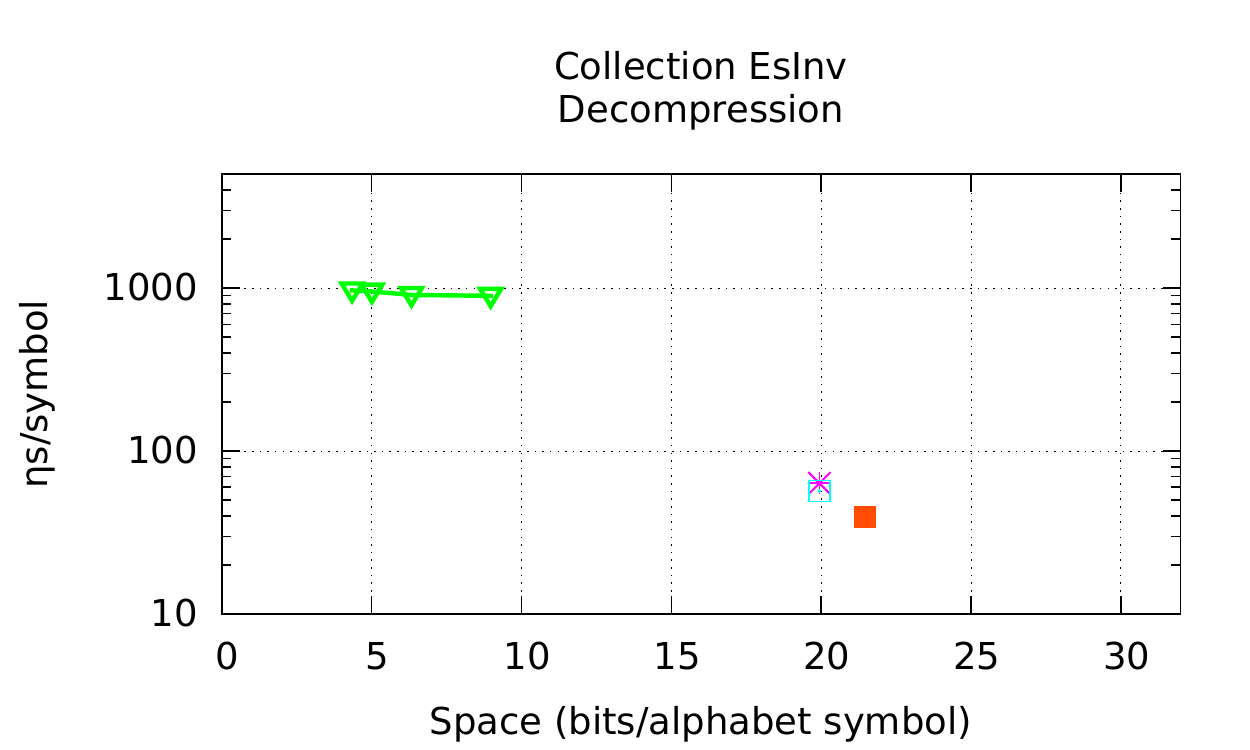}}
\end{tabular}

\begin{tabular}{cc}
    \subfigure{ \label{resultf}
    \includegraphics[width=0.49\textwidth]{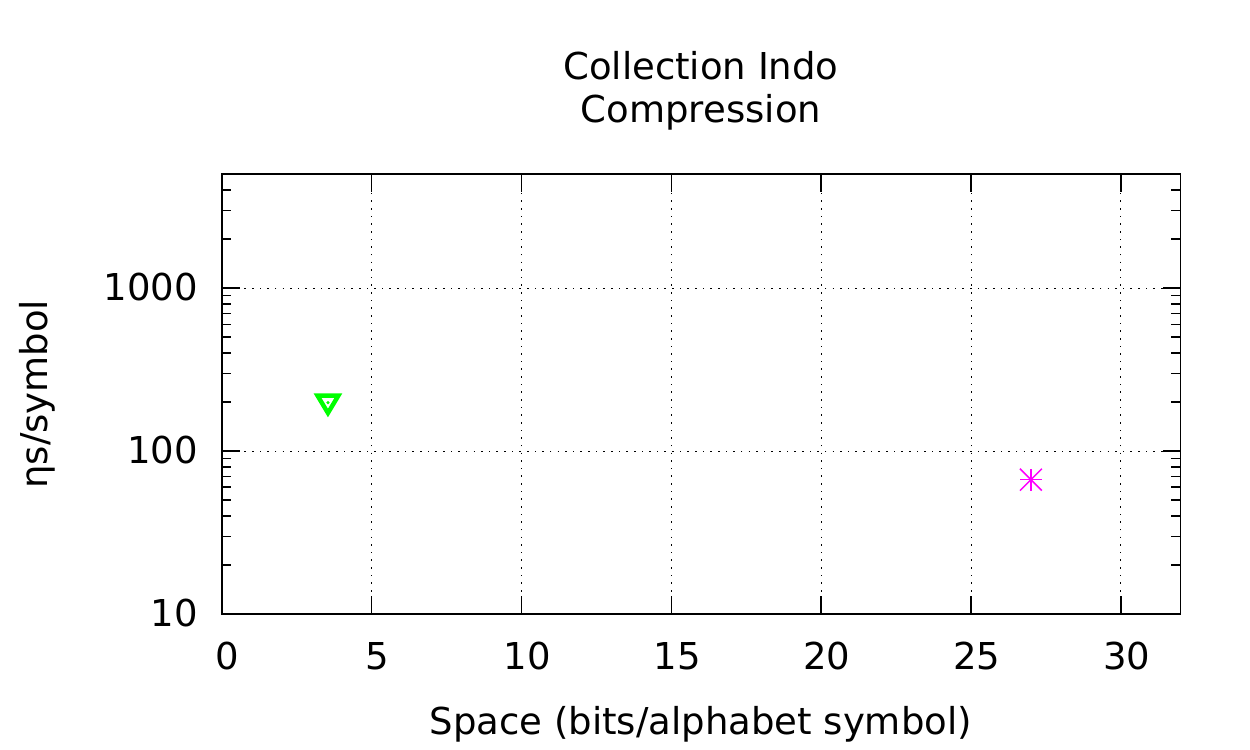}}
&
    \subfigure{ \label{resultc}
    \includegraphics[width=0.49\textwidth]{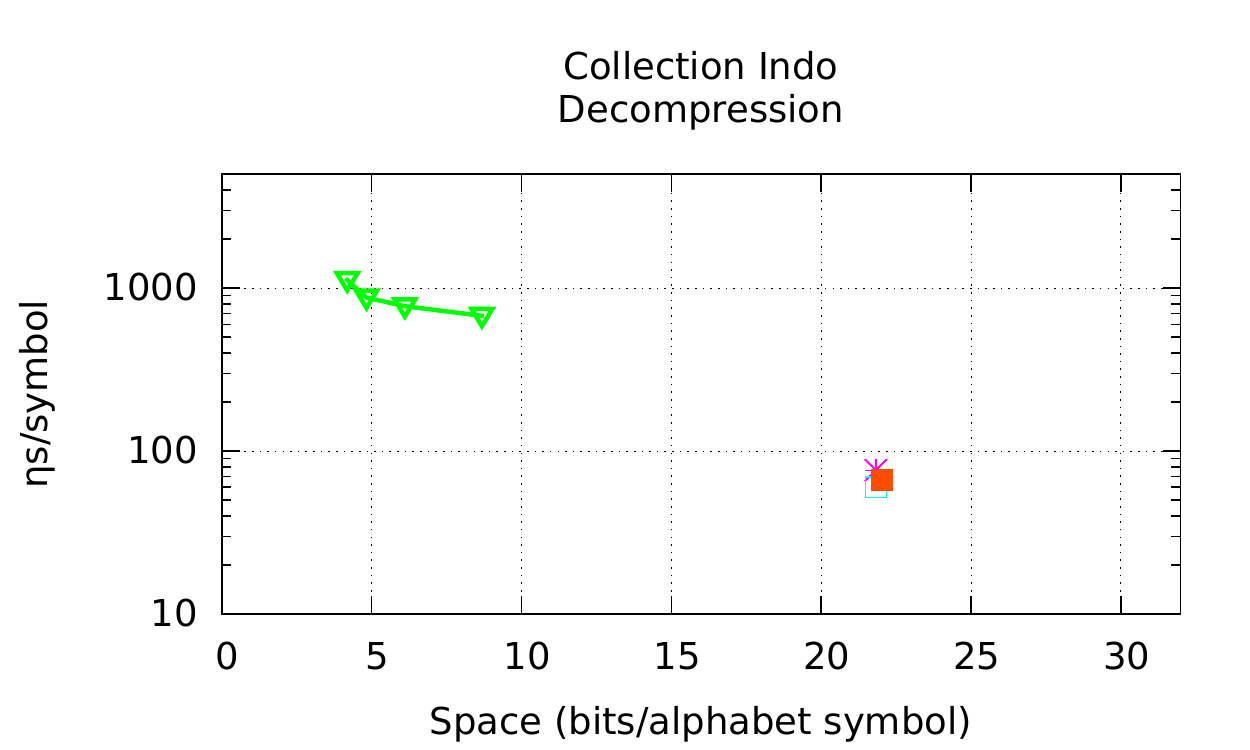}}
\end{tabular}

    \end{center}

   \caption{Code representation size versus compression/decompression time for
table based representations (TABLE, TABLE$_S$, and TABLE$_E$) and ours (COMPR). 
Time (in logscale) is measured in nanoseconds per symbol.} 
\label{modelos.optimal}
\end{figure}

It can be seen that our compressed representations takes just around 12\% of the space of the table implementation for compression
(an 8-fold reduction),
while being 2.5--8 times slower. Note that
compression is performed by carrying out $rank$ operations on the wavelet tree bitmaps. 
Therefore, we do not consider the space overhead incurred to speed up $select$
operations, and we only plot a single point for technique COMPR at compression charts. 
Also, we only show the simple (and most compact) TABLE variant, as the improvements of the others
apply only to decompression.

For decompression, our solution (COMPR) takes 17\% to 45\% of the space of the 
TABLE$_*$ variants (thus reaching almost a 6-fold space reduction), but it is 
also 12--24 times slower.
This is because our solution uses operation $select$ for decompression, and
this is slower than $rank$ even with the structures for speeding it up.

Overall, our compact representation is able to compress at a rate around 
2.5--5 MB/sec and decompress at 1 MB/sec, while using much less space than 
a classical Huffman implementation (which compresses/decompresses at around
14--25 MB/sec).

Finally, note that we only need a single data structure to both compress and 
decompress, while the naive approach uses different tables for each operation. 
In the cases where both functionalities are simultaneously necessary (as in
compressed sequence representations \cite{Nav14}), our structure 
uses as little as 7\% of the space needed by a classical representation.

\subsection{Length-Limited Codes}

In the theoretical description, we refer to an optimal technique for limiting 
the length of the code trees to a given value $\lmax \ge \lceil \lg n \rceil$ 
\cite{LH90}, as well as several heuristics and approximations:

\begin{itemize}
	\item \texttt{Milidi\'u}: the approximate technique proposed by
Milidi\'u and Laber \cite{ML01} that nevertheless guarantees the upper bound we 
have used in the paper.  It takes $\Oh{n}$ time.

	\item  \verb|Increase|: inspired in the bounds of Katona and
Nemetz~\cite{KN76}, we start with $f=2$ and set to $f$ the frequency
of each symbol whose frequency is $<f$. Then we build the Huffman tree, and 
if its height is $\leq \lmax$, we are done. Otherwise, we increase $f$ by 
$1$ and repeat the process. Since the Huffman construction algorithm is
linear-time once the symbols are sorted by frequency and the process does not
need to reorder them, this method takes $\Oh{n\log (n\phi^{-\lmax})}
= \Oh{n\log n}$ time if we use 
exponential search to find the correct $f$ value. A close predecessor of this
method appears in Chapter 9 of {\em Managing Gigabytes} \cite{WMB99}. They use 
a multiplicative instead of an additive approximation, so as to find an 
appropriate $f$ faster. Thus they may find a value of $f$ that is larger than
the optimal.

	\item \verb|Increase-A|: analogous to \verb|Increase|, but instead
adds $f$ to the frequency of each symbol. 

	\item \verb|Balance|: the technique (e.g., see \cite{BN09}), 
that balances the parents
of the maximal subtrees that, even if balanced, exceed the maximum allowed
height. It also takes $\Oh{n}$ time. In the case of a canonical Huffman tree, this
is even simpler, since only one node along the rightmost path of the tree
needs to be balanced.

	\item \verb|Optimal|: the package-merge algorithm of Larmore and
Hirshberg \cite{LH90}. Its time complexity is $\Oh{n\,\lmax}$.

%
%
\end{itemize}


\begin{figure}[p]
\begin{center}
\includegraphics[width=0.55\textwidth]{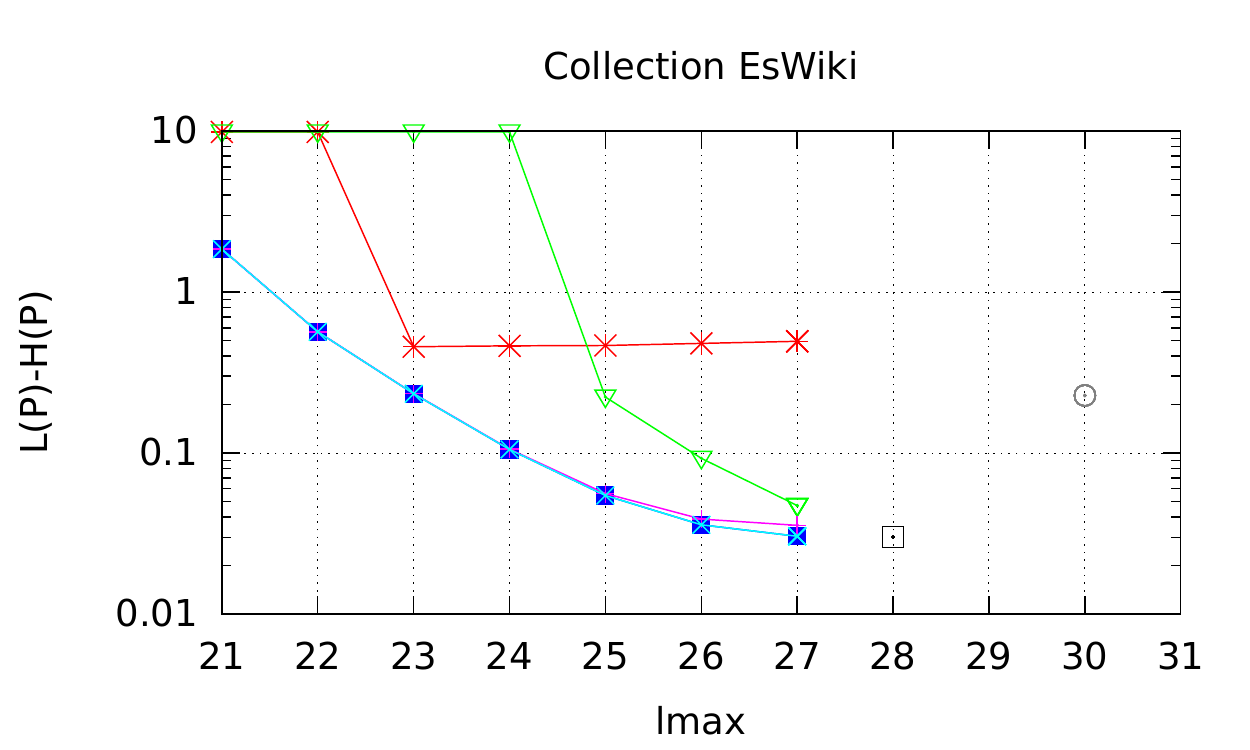} 
\includegraphics[width=0.55\textwidth]{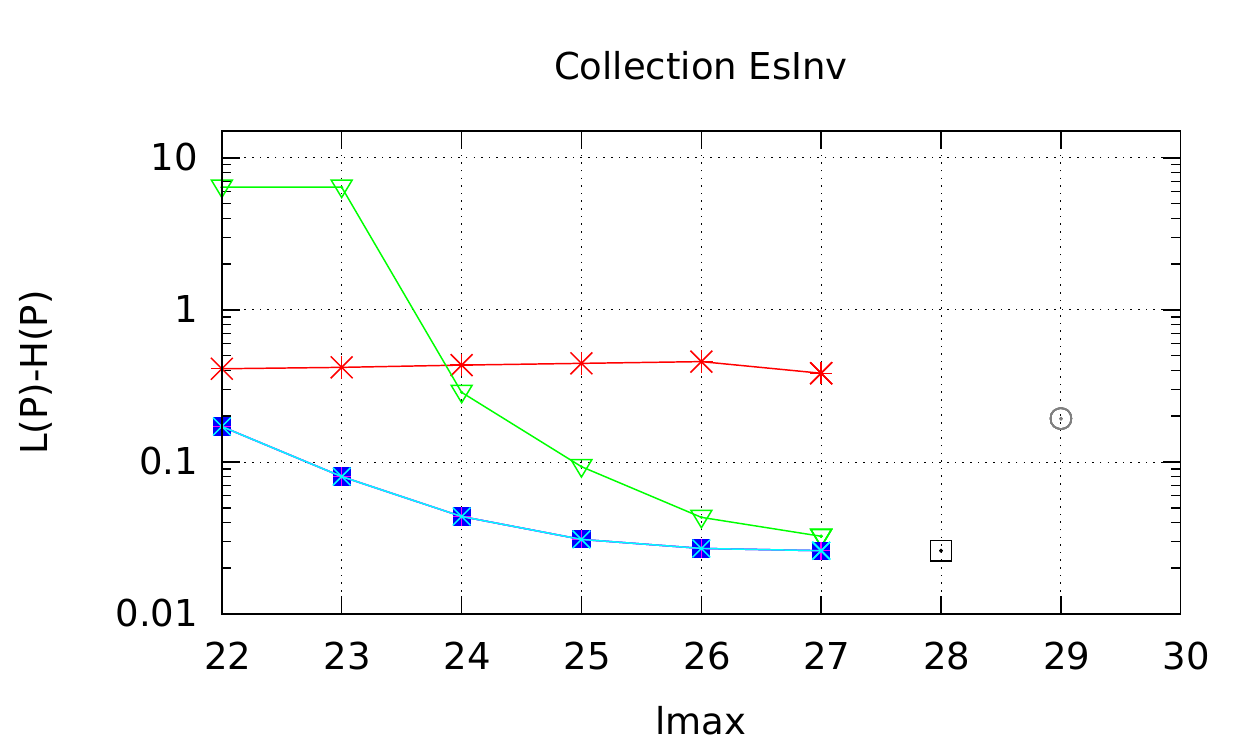} 
\includegraphics[width=0.55\textwidth]{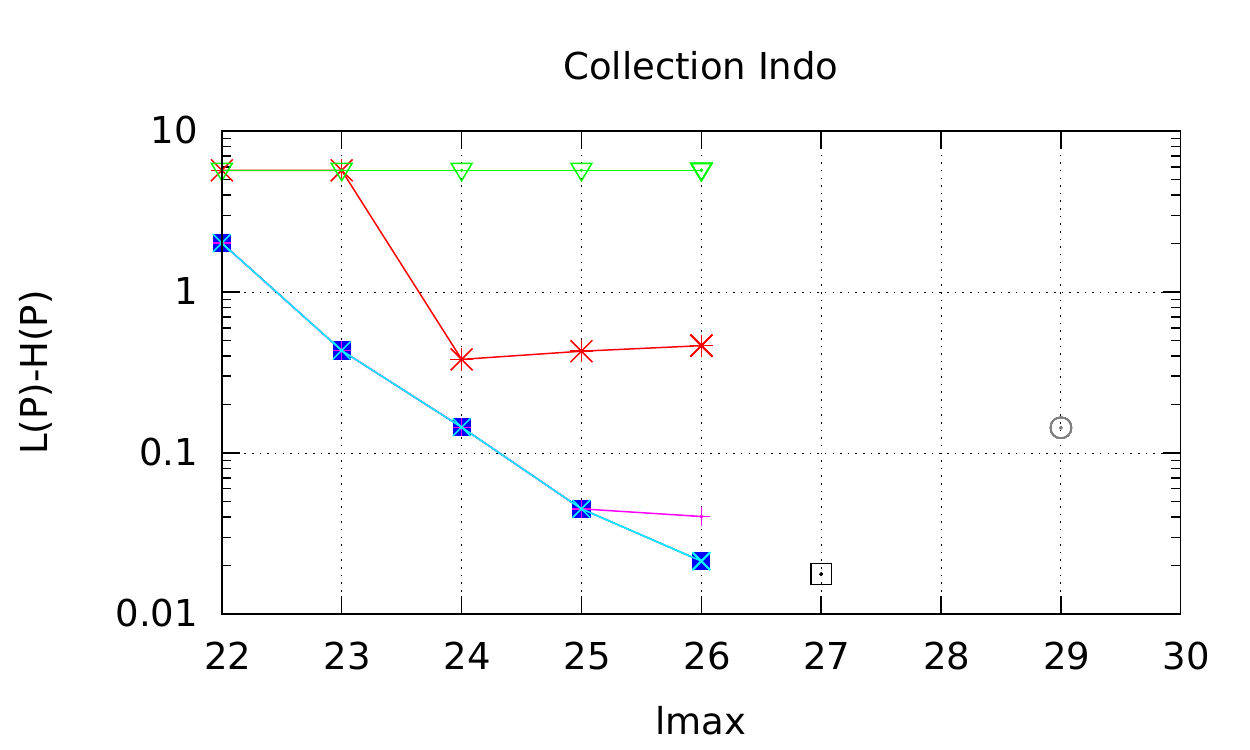} 
\begin{minipage}{0.45\textwidth}
\vspace{5mm}
\centerline{
\includegraphics[width=0.8\textwidth]{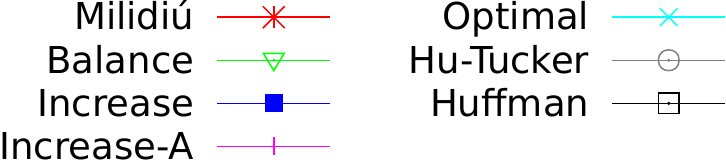}
}
\end{minipage}
\end{center}
\caption{Comparison of the length-restricted approaches measured as their
additive redundancy (in logscale) over the zero-order empirical entropy, $\HH(P)$, for each value of
$\lmax$. We also include Hu-Tucker and Huffman as reference points.}
\label{fig:suboptimal}
\end{figure}

Figure~\ref{fig:suboptimal} compares the techniques for all the meaningful
$\lmax$ values, showing the additive redundancy they produce over $\HH(P)$.
It can be seen that the average code lengths obtained by 
\texttt{Milidi\'u}, although they have theoretical guarantees, are not so good
in practice. They are comparable with those of \verb|Balance|, a simpler 
and still linear-time heuristic, which however
does not provide any guarantee and sometimes can only return a completely
balanced tree. On the other hand, technique \verb|Increase| performs better
than or equal to \verb|Increase-A|, and actually matches the average code 
length of \verb|Optimal| systematically in the three collections. 

Techniques \texttt{Milidi\'u}, \verb|Balance|, and \verb|Optimal| are all
equally fast in practice, taking about 2 seconds to find their
length-restricted code in our collections. 
The time for  \verb|Increase| and \verb|Increase-A|
depends on the value of $\lmax$. For large values of $\lmax$, they also take
around 2 seconds, but this raises up to 20 seconds when $\lmax$ is closer to 
$\lceil \lg n \rceil$ (and thus the value $f$ to add is larger, up to 100--300 
in our sequences).

In practice, technique \verb|Increase| can be recommended for its extreme
simplicity to implement and very good approximation results. If the
construction time is an issue, then \verb|Optimal| should be used. It
performs fast in practice and it is not so hard to implement%
\footnote{There are even some public implementations, for example
{\tt https://gist.github.com/imaya/3985581}}.
For the following experiments, we will use the results of
\verb|Optimal|/\verb|Increase|.

As a final note, observe that by restricting the code length to, say, 
$\lmax=22$ on \EsWiki\ and \EsInv\, and $\lmax=23$ on \Indo, the additive 
redundancy obtained is below $\epsilon = 0.6$, and the redundancy is below 
5\% of $\HH(P)$.



\subsection{Approximations}

Now we evaluate the additive and multiplicative approximations, in terms of 
average code length $\LL$, compression and decompression performance.
We compare them with two optimal model representations, OPT-T and OPT-C, which correspond
to TABLE and COMPR of Section~\ref{sec:opt-exp}. The additive approximations
(Section~\ref{sec:additive}) included, ADD+T and ADD+C, are obtained by
restricting the maximum code lengths to $\lmax$ and storing the resulting 
codes using TABLE or COMPR, respectively. We show one point per $\lmax = 
22\ldots 27$ on \EsWiki\ and \EsInv, and $\lmax = 22\ldots 26$ on \Indo. 
For the multiplicative approximation (Section~\ref{sec:multiplicative}), we 
test the variants MULT-$\lmax$, which limit $\lmax$ to 25 and 26, and use 
$c$ values 1.5, 1.75, 2, and 3. For all the solutions that use a wavelet 
tree, we have fixed a $select$ sampling rate to $32$.

Figure \ref{fig:modelos.approx.model.optimality} shows the results in terms of
bps for storing the model versus the resulting redundancy of
the code, measured as $\LL(P)/\HH(P)$.

\begin{figure}[p]
    \begin{center}
    \includegraphics[width=0.55\textwidth]{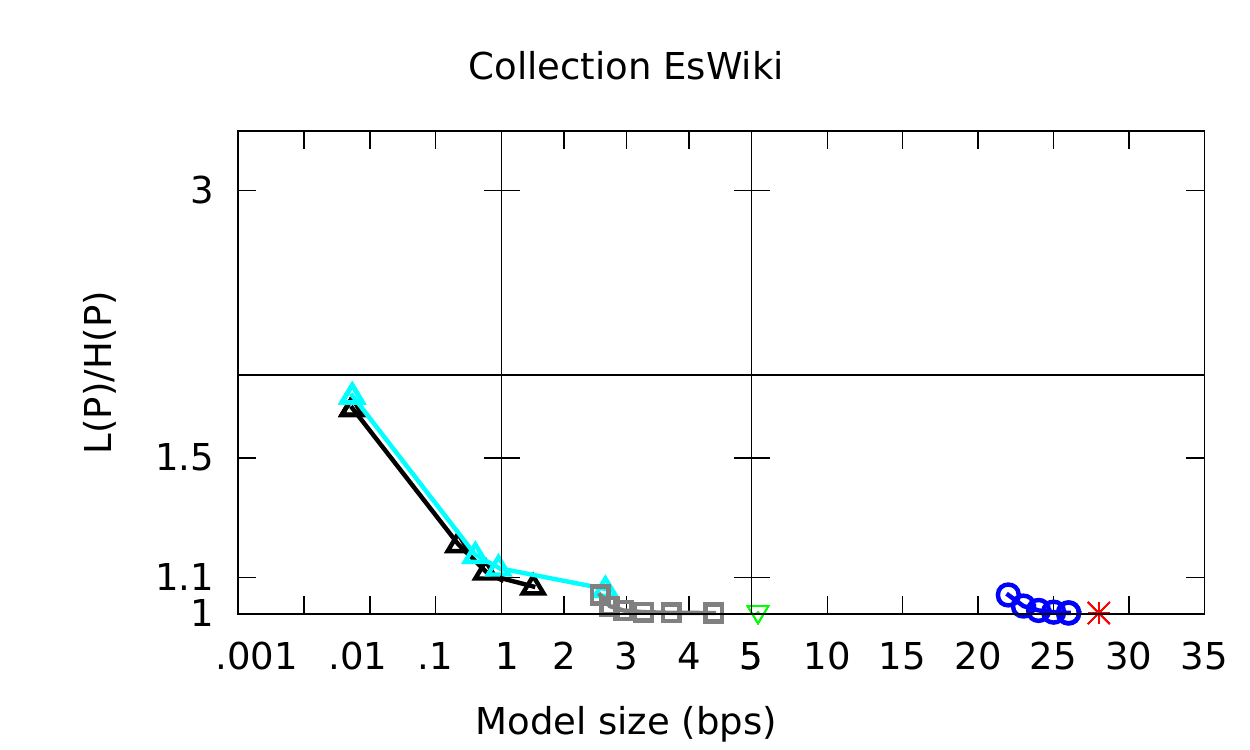}
    \includegraphics[width=0.55\textwidth]{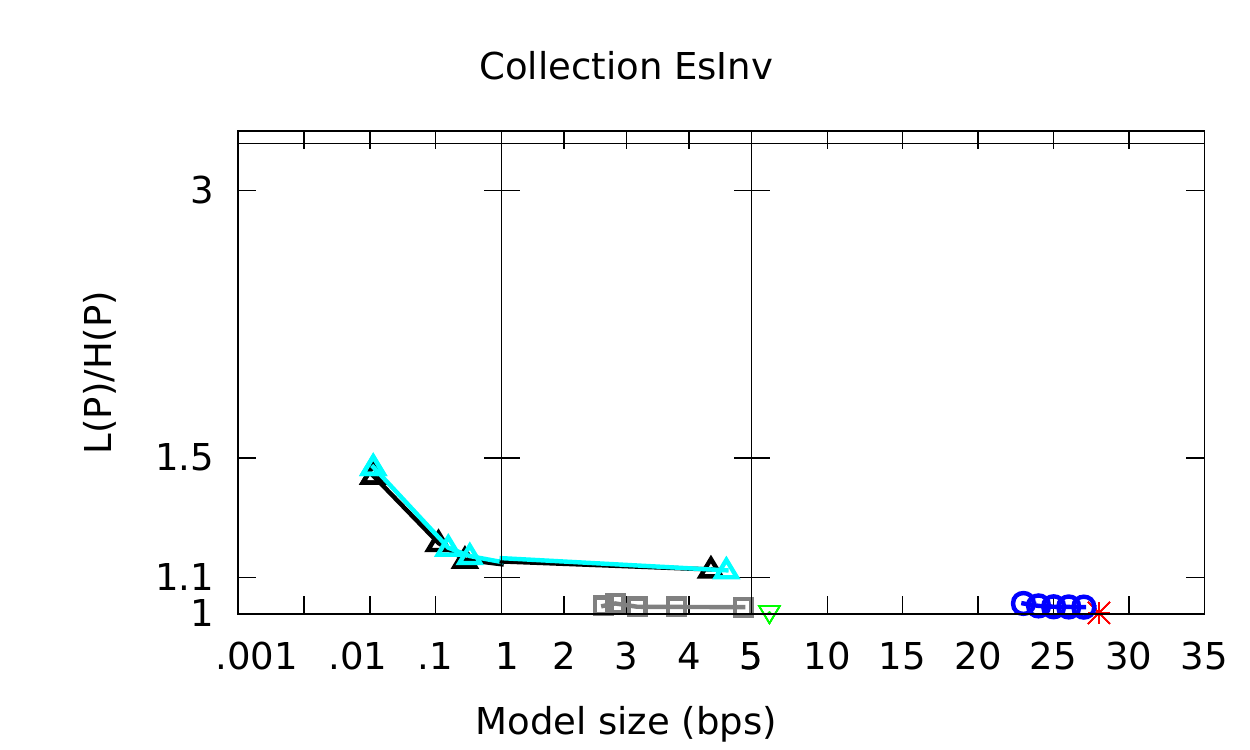}
    \includegraphics[width=0.55\textwidth]{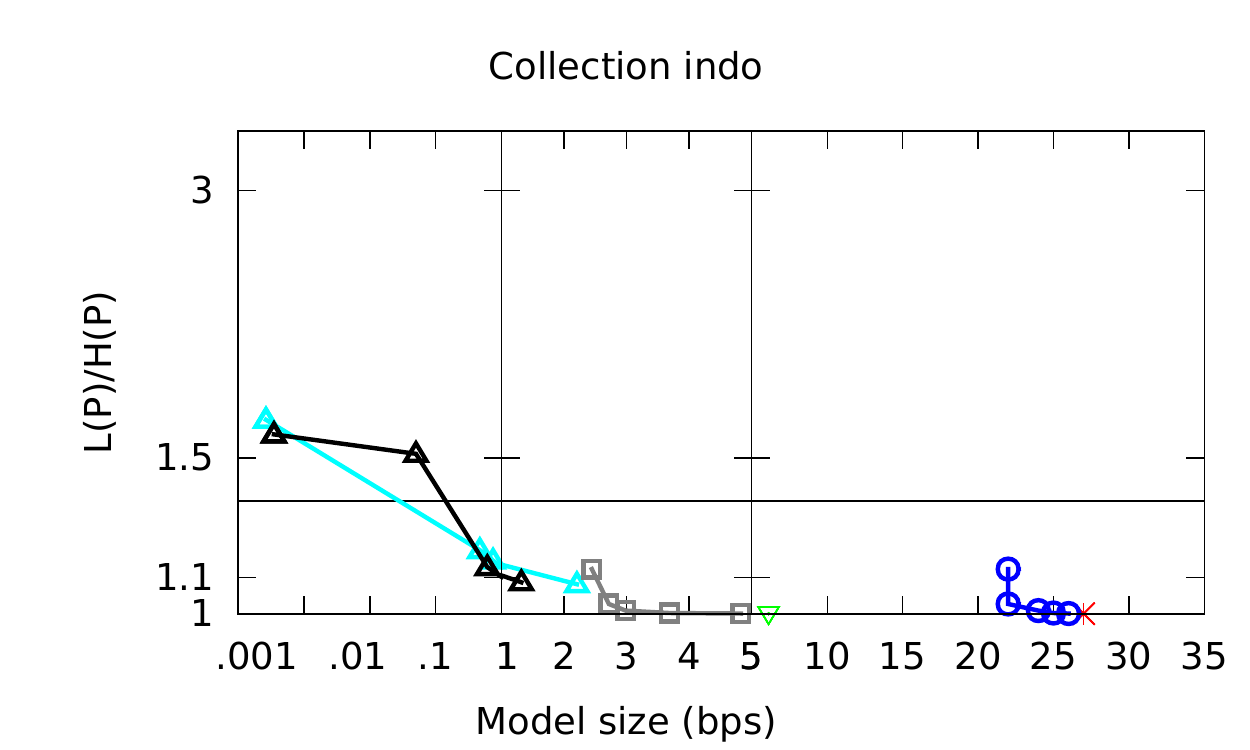}
\begin{minipage}{0.45\textwidth}
\vspace{5mm}
\centerline{
    \includegraphics[width=0.8\textwidth]{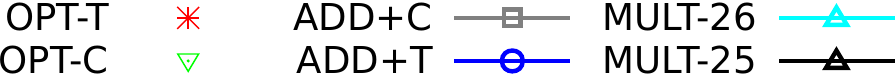}
}
\end{minipage}
    \end{center}
   \caption{Relation between the size of the model and the average code length. 
  The $x$-axes are in logscale for values smaller than 1 and in two linear
scales for 1--5 and 5--35. The horizontal line
shows the limit $\lceil \lg n\rceil / \HH(P)$, where no compression is
obtained compared with a fixed-length code.}
\label{fig:modelos.approx.model.optimality}
\end{figure}

The additive approximations have a mild impact when implemented in classical
form. However, the compact representation, ADD+C, reaches half the space of our
exact compact representation, OPT-C. This is obtained at
the price of a modest redundancy, below 5\% in all cases, if one uses reasonable
values for $\lmax$. 

With the larger $c$ values, the multiplicative approach is extremely efficient 
for storing the model, reaching reductions up to 2 and 3 orders of magnitude
with respect to the classic representations. However, this comes at the price
of a redundancy that can reach 50\%. The redundancy may go beyond 
$\lceil \lg n\rceil / \HH(P)$, at which point it is better to use a plain code 
of $\lceil \lg n\rceil$ bits. Instead, with value $c=1.75$, the model size is
still 20 times smaller than a classical representation, and 2--3 times smaller
than the most compact representation of additive approximations, with a redundancy 
only slightly over 10\%.

Figure \ref{fig:modelos.approx.code.decode} compares these representations in 
terms of compression and decompression performance. The numbers near each 
point show the redundancy (as a percentage over the entropy) of
the model representing that point. We use ADD+C with values
$\lmax=22$ on \EsWiki\ and \EsInv\, and $\lmax=23$ on \Indo. For
ADD+T, the decompression times are the same for all the tested $\lmax$ values. In this
figure we set the $select$ samplings of the wavelet trees to $(32, 64, 128)$.
We also include in the comparison the variant MULT-26 with $c=1.75$ and 1.5.

It can be seen that the multiplicative approach is very fast, comparable to the
table-based approaches ADD+T and OPT-T: 10\%--50\% slower at compression and
at most 20\% slower at decompression. Within this speed, if we use $c=1.75$,
the representation is 6--11 times smaller than the classical one for 
compression and 5--9 times for decompression, at the price of about 10\% of
redundancy. If we choose $c=1.5$, the redundancy increases to about 20\% but
the model becomes an order of magnitude smaller.

The compressed additive approach (ADD+C) achieves a smaller model than the
multiplicative one with $c=1.75$ (it is 14 times smaller than the classical
representation for compression and 11 times for decompression). This is 
achieved with significantly less redundancy than the multiplicative model, 
just 3\%--5\%. However, although ADD+C is about 20\%--30\% faster than the 
exact code OPT-C, it is still significantly slower than the table-based 
representations (2--5.5 times slower for compression and 9--17 for decompression).

\begin{figure}[p]

    \begin{center}

\begin{tabular}{cc}
    \subfigure{ \label{resulta2}
    \includegraphics[width=0.49\textwidth]{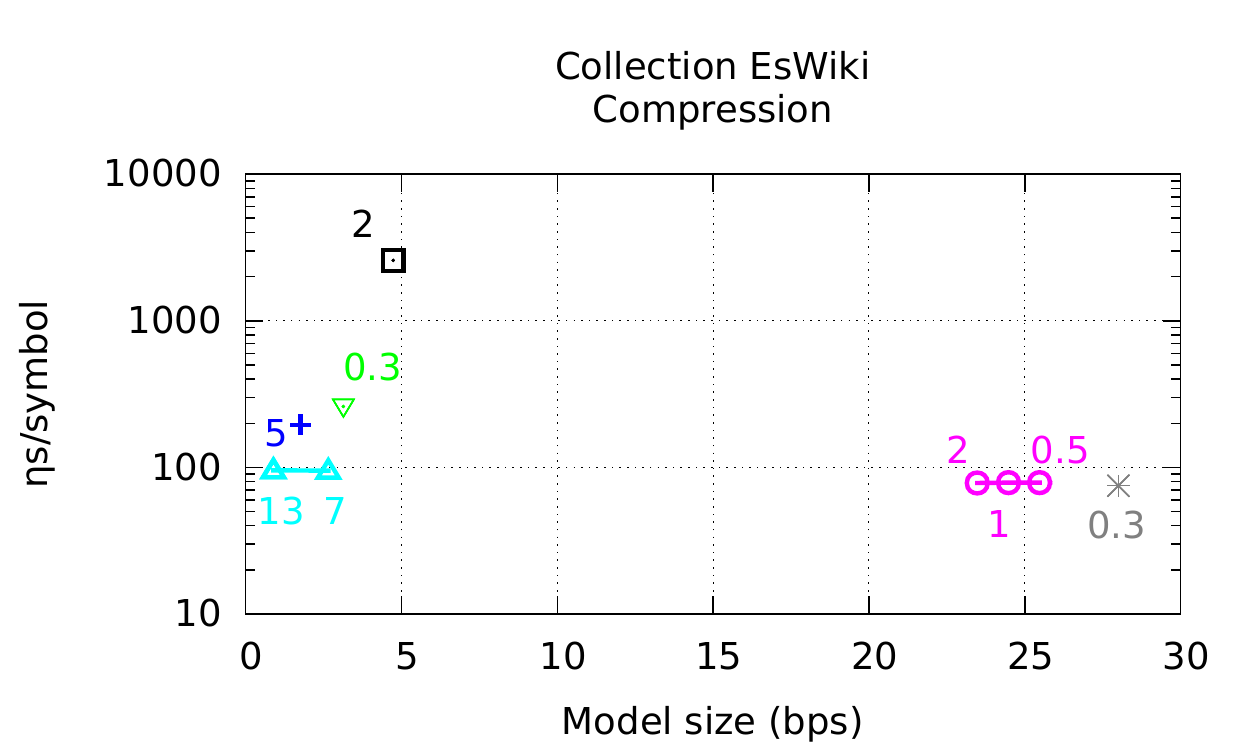}}
&
    \subfigure{ \label{resultd2}
    \includegraphics[width=0.49\textwidth]{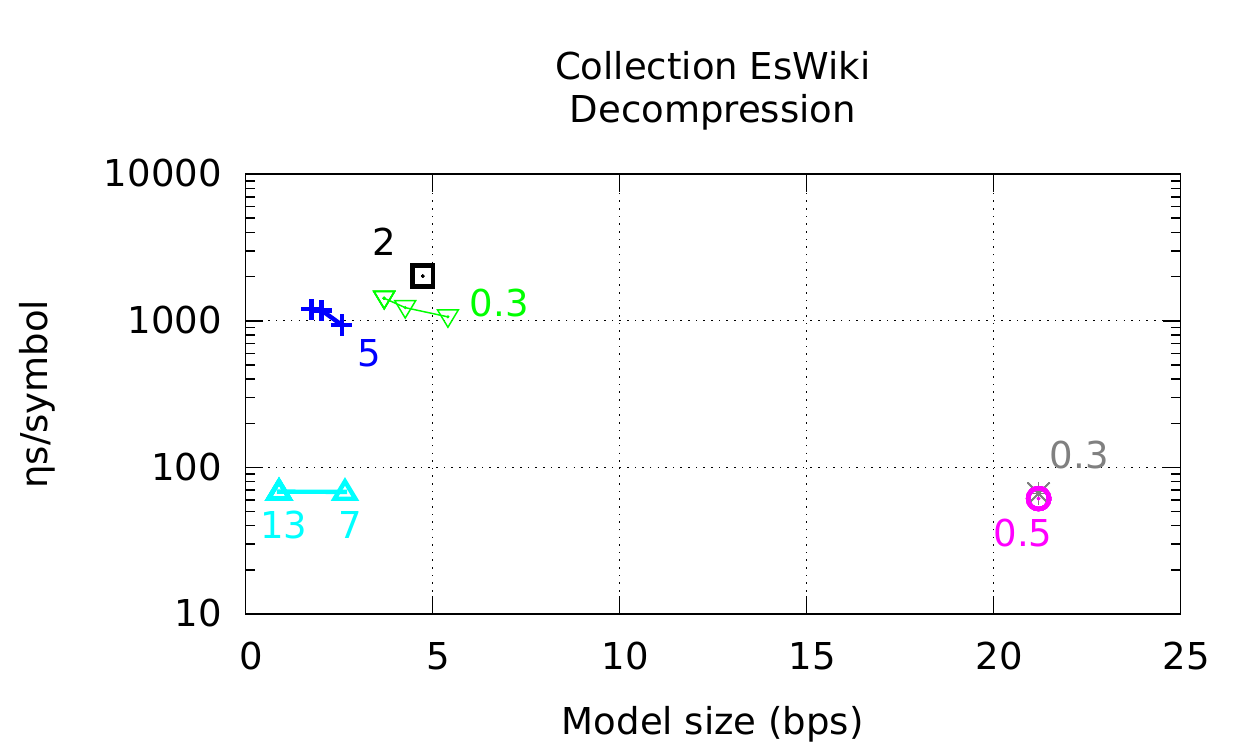}}
\end{tabular}

\begin{tabular}{cc}
    \subfigure{ \label{resultb2}
    \includegraphics[width=0.49\textwidth]{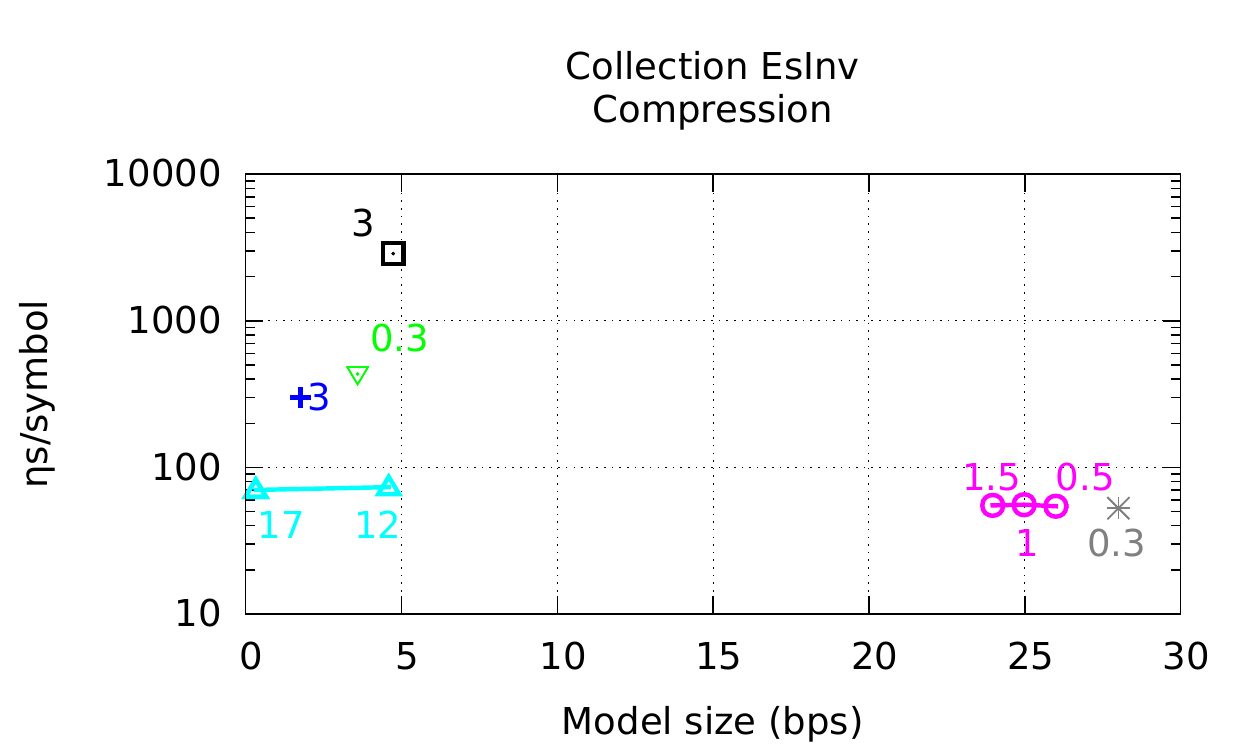}
    }
&
    \subfigure{ \label{resulte2}
    \includegraphics[width=0.49\textwidth]{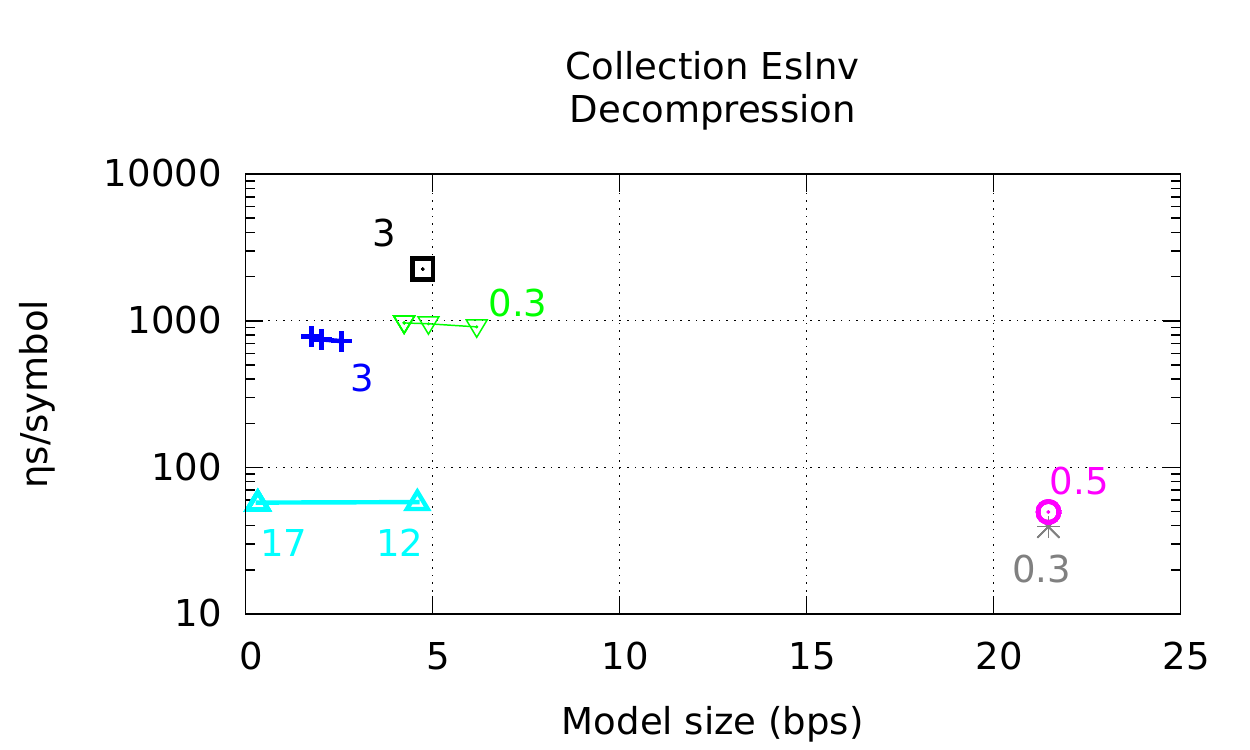}}
\end{tabular}

\begin{tabular}{cc}
    \subfigure{ \label{resultc2}
    \includegraphics[width=0.49\textwidth]{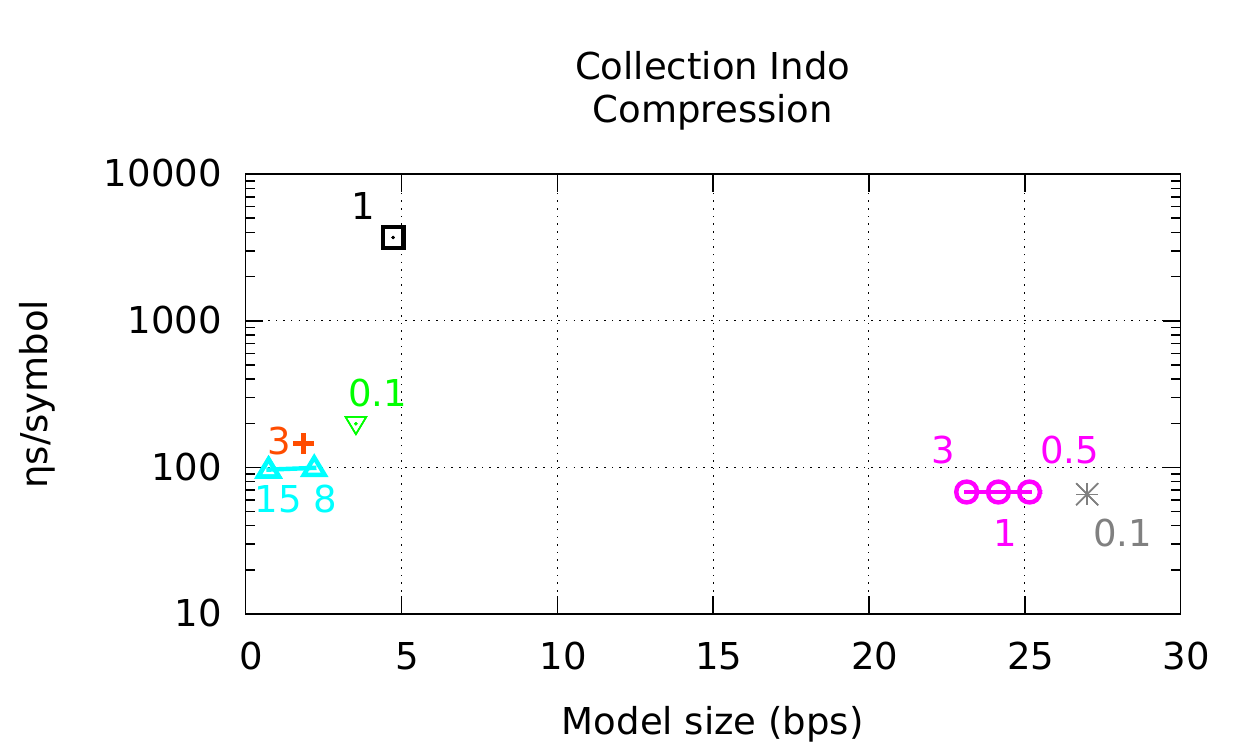}}
&
    \subfigure{ \label{resultf2}
    \includegraphics[width=0.49\textwidth]{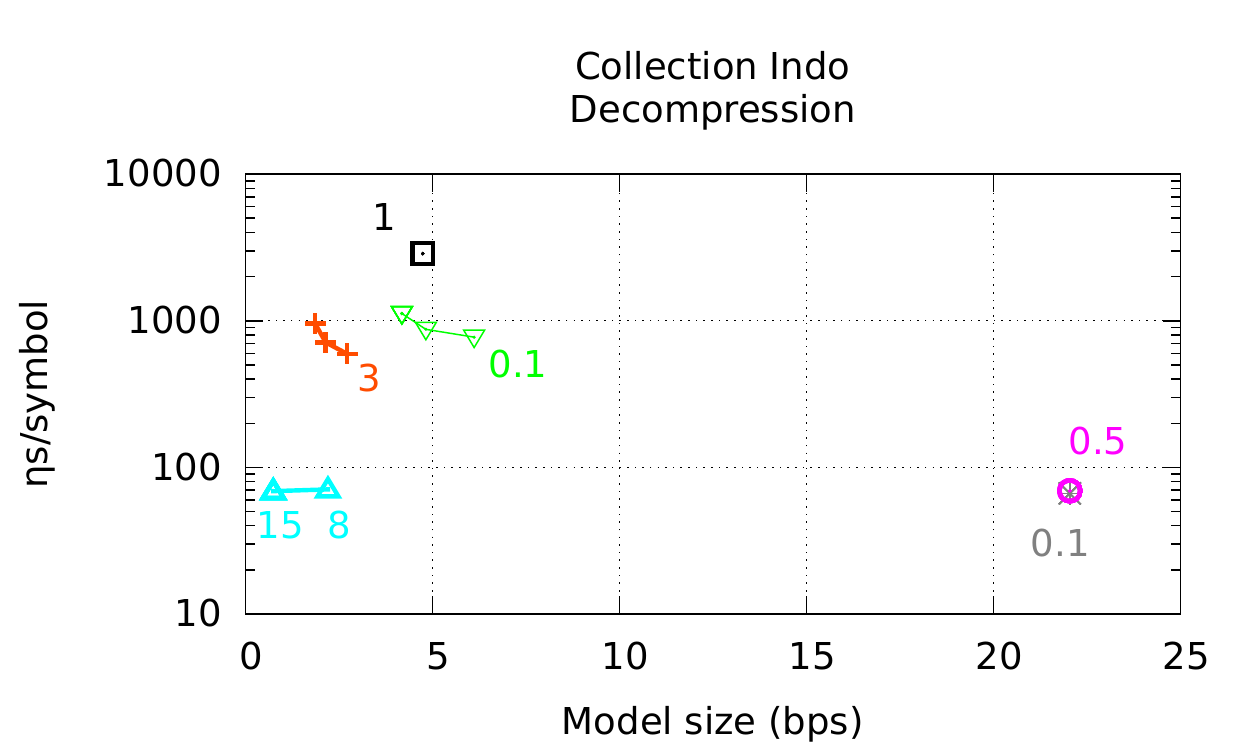}}
\end{tabular}

\includegraphics[scale=0.75]{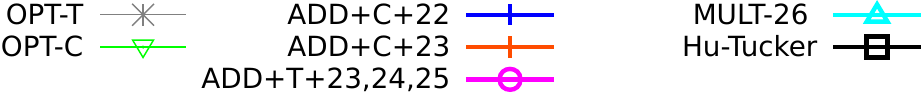}

    \end{center}
   \caption{Space/time performance of the approximate and exact approaches. 
Times are in
nanoseconds per symbol and in logscale. The numbers around the points are
their redundancy as a percentage of the entropy.}
\label{fig:modelos.approx.code.decode}
\end{figure}

Finally, we can see that our compact implementation of Hu-Tucker codes
achieves competitive space, but it is an order of magnitude slower than
our additive approximations, which can always use simultaneously less space 
and time. With respect to the redundancy, Figure~\ref{fig:suboptimal} shows 
that Hu-Tucker codes are equivalent to our additive approximations with 
$\lmax=23$ on \EsWiki, $\lmax=22$ on \EsInv, and $\lmax=24$ on \Indo. This 
shows that the use of alphabetic codes as a suboptimal code to reduce the 
model representation size is inferior, in all aspects, to our additive 
approximations. Figure~\ref{fig:modelos.approx.code.decode} shows that
Hu-Tucker is also inferior, in the three aspects, to our compact optimal
codes, OPT-C. We remark that alphabetic
codes are interesting by themselves for other reasons, in contexts where
preserving the order of the source symbols is important.

\section{Conclusions}

We have explored the problem of providing compact representations of Huffman
models. The model size is relevant in several applications, particularly
because it must reside in main memory for efficient compression and
decompression.

We have proposed new representations achieving constant compression and
decompression time per symbol while using $\Oh{n\log\log(N/n)}$ bits per
symbol, where $n$ is the alphabet size and $N$ the sequence length. This is in
contrast to the (at least) $\Oh{n\log n}$ bits used by previous representations. In
our practical implementation, the time complexities are
$\Oh{\log\log N}$ and even $\Oh{\log N}$,
but we do achieve 8-fold space reductions for compression and up to 
6-fold for decompression. This comes, however, at the price of increased 
compression and decompression time (2.5--8 times slower at compression and 12--24
at decompression), compared to current representations. In low-memory
scenarios, the space reduction can make the difference between fitting the 
model in main memory or not, and thus the increased times are the price to
pay.

We also showed that, by tolerating a small additive overhead of $\epsilon$ on
the average code length, the model can be stored in
$\Oh{n\log\log(1/\epsilon)}$ bits, while maintaining constant compression and
decompression time. In practice, these additive approximations can halve 
our compressed model size (becoming 11--14 times smaller than a classical
representation), while incurring a very small increase (5\%) in the average 
code length. They are also faster, but still 2--5.5 times slower for 
compression and 5--9 for decompression.

Finally, we showed that a multiplicative penalty in the average code length
allows storing the model in $o(n)$ bits. In practice, the reduction in model
size is sharp, while the compression and decompression times are only 10\%--50\%
and 0\%--20\% slower, respectively, than classical implementations. Redundancies
are higher, however. With 10\% of redundancy, the model size is close to
that of the additive approach, and with 20\% the size decreases by 
another order of magnitude.

Some challenges for future work are:
\begin{itemize}
\item Adapt these representations to dynamic scenarios, where the model 
undergoes changes as compression/decompression progresses. While our compact
representations can be adapted to support updates, the main problem is how to
efficiently maintain a dynamic canonical Huffman code. We are not aware of
such a technique.
\item Find more efficient representations of alphabetic codes. Our baseline
achieves reasonably good space, but the navigation on the compact tree
representations slows it down considerably. It is possible that faster
representations supporting left/right child and subtree size can be found.
\item Find constant-time encoding and decoding methods that are fast and
compact in practice. Multiary wavelet trees \cite{Bow10} are faster than
binary wavelet trees, but generally use much more space. Giving them the
shape of a (multiary) Huffman tree and using plain representations for the
sequences in the nodes could reduce the space gap with our binary Huffman-%
shaped wavelet trees used to represent $L$. As for the fusion trees, looking
for a practical implementation of trees with arity $w^\epsilon$, which
outperforms a plain binary search, is interesting not only for this problem, 
but in general for predecessor searches on small universes.
\end{itemize}

\paragraph{Acknowledgements.} We thank the reviewers, whose comments
helped improve the paper significantly.

\bibliographystyle{plain}
\bibliography{paper}

\begin{thebibliography}{10}

\bibitem{AM01}
M.~Adler and B.~M. Maggs.
\newblock Protocols for asymmetric communication channels.
\newblock {\em Journal of Computer and System Sciences}, 63(4):573--596, 2001.

\bibitem{ACNS10}
D.~Arroyuelo, R.~C{\'a}novas, G.~Navarro, and K.~Sadakane.
\newblock Succinct trees in practice.
\newblock In {\em Proc. 11th Workshop on Algorithm Engineering and Experiments
  (ALENEX)}, pages 84--97, 2010.

\bibitem{BN09}
J.~Barbay and G.~Navarro.
\newblock Compressed representations of permutations, and applications.
\newblock In {\em Proc. 26th International Symposium on Theoretical Aspects of
  Computer Science (STACS)}, pages 111--122, 2009.

\bibitem{BN13}
J.~Barbay and G.~Navarro.
\newblock On compressing permutations and adaptive sorting.
\newblock {\em Theoretical Computer Science}, 513:109--123, 2013.

\bibitem{BN12}
D.~Belazzougui and G.~Navarro.
\newblock New lower and upper bounds for representing sequences.
\newblock In {\em Proc. 20th Annual European Symposium on Algorithms (ESA)},
  LNCS 7501, pages 181--192, 2012.

\bibitem{BSTW86}
J.~L. Bentley, D.~D. Sleator, R.~E. Tarjan, and V.~K. Wei.
\newblock A locally adaptive data compression scheme.
\newblock {\em Communications of the ACM}, 29(4), 1986.

\bibitem{Bow10}
A.~Bowe.
\newblock {\em Multiary Wavelet Trees in Practice}.
\newblock Honours thesis, RMIT University, Australia, 2010.

\bibitem{BFLN12}
N.~Brisaboa, A.~Fari{\~n}a, S.~Ladra, and G.~Navarro.
\newblock Implicit indexing of natural language text by reorganizing bytecodes.
\newblock {\em Information Retrieval}, 15(6):527--557, 2012.

\bibitem{BFNP07}
N.~Brisaboa, A.~Fari{\~n}a, G.~Navarro, and J.~Param\'a.
\newblock Lightweight natural language text compression.
\newblock {\em Information Retrieval}, 10:1--33, 2007.

\bibitem{BNO12}
N.~Brisaboa, G.~Navarro, and A.~Ord{\'o}{\~n}ez.
\newblock Smaller self-indexes for natural language.
\newblock In {\em Proc. 19th International Symposium on String Processing and
  Information Retrieval (SPIRE)}, LNCS 7608, pages 372--378, 2012.

\bibitem{Bur93}
M.~Buro.
\newblock On the maximum length of {H}uffman codes.
\newblock {\em Information Processing Letters}, 45(5):219--–223, 1993.

\bibitem{CK85}
Y.~Choueka, S.~T. Klein, and Y.~Perl.
\newblock Efficient variants of {H}uffman codes in high level languages.
\newblock In {\em Proc. 8th Annual International ACM Conference on Research and
  development in Information Retrieval (SIGIR)}, pages 122--130. ACM, 1985.

\bibitem{CN10}
F.~Claude and G.~Navarro.
\newblock Fast and compact {W}eb graph representations.
\newblock {\em ACM Transactions on the Web}, 4(4):article 16, 2010.

\bibitem{DRR12}
P.~Davoodi, R.~Raman, and S.~Rao Satti.
\newblock Succinct representations of binary trees for range minimum queries.
\newblock In {\em Proc. 18th Annual International Conference on Computing and
  Combinatorics (COCOON)}, LNCS 7434, pages 396--407, 2012.

\bibitem{FGMS05}
P.~Ferragina, R.~Giancarlo, G.~Manzini, and M.~Sciortino.
\newblock Boosting textual compression in optimal linear time.
\newblock {\em Journal of the ACM}, 52(4):688--713, 2005.

\bibitem{FMMN07}
P.~Ferragina, G.~Manzini, V.~M{\"a}kinen, and G.~Navarro.
\newblock Compressed representations of sequences and full-text indexes.
\newblock {\em ACM Transactions on Algorithms}, 3(2):article 20, 2007.

\bibitem{FKS84}
M.~L. Fredman, J.~Koml\'os, and E.~Szemer\'edi.
\newblock Storing a sparse table with {$\Oh{1}$} worst case access time.
\newblock {\em Journal of the ACM}, 31(3):538--544, 1984.

\bibitem{FW93}
M.~L. Fredman and D.~E. Willard.
\newblock Surpassing the information theoretic bound with fusion trees.
\newblock {\em Journal of Computer and System Sciences}, 47(3):424--436, 1993.

\bibitem{Gag06a}
T.~Gagie.
\newblock Compressing probability distributions.
\newblock {\em Information Processing Letters}, 97(4):133--137, 2006.

\bibitem{Gag06b}
T.~Gagie.
\newblock Large alphabets and incompressibility.
\newblock {\em Information Processing Letters}, 99(6):246--251, 2006.

\bibitem{Gag08}
T.~Gagie.
\newblock Dynamic asymmetric communication.
\newblock {\em Information Processing Letters}, 108(6):352--355, 2008.

\bibitem{GNN10}
T.~Gagie, G.~Navarro, and Y.~Nekrich.
\newblock Fast and compact prefix codes.
\newblock In {\em Proc. 36th International Conference on Current Trends in
  Theory and Practice of Computer Science (SOFSEM)}, LNCS 5901, pages 419--427,
  2010.

\bibitem{GN09}
T.~Gagie and Y.~Nekrich.
\newblock Worst-case optimal adaptive prefix coding.
\newblock In {\em Proc. 9th Symposium on Algorithms and Data Structures
  (WADS)}, pages 315--326, 2009.

\bibitem{GM59}
E.~N. Gilbert and E.~F. Moore.
\newblock Variable-length binary encodings.
\newblock {\em Bell System Technical Journal}, 38:933--967, 1959.

\bibitem{GGV03}
R.~Grossi, A.~Gupta, and J.~Vitter.
\newblock High-order entropy-compressed text indexes.
\newblock In {\em Proc. 14th ACM-SIAM Symposium on Discrete Algorithms (SODA)},
  pages 841--850, 2003.

\bibitem{HTOC95}
R.~Hashemian.
\newblock Memory efficient and high-speed search {H}uffman coding.
\newblock {\em Communications, IEEE Transactions on}, 43(10):2576--2581, 1995.

\bibitem{HT71}
T.~C. Hu and A.~C. Tucker.
\newblock Optimal computer search trees and variable-length alphabetical codes.
\newblock {\em SIAM Journal of Applied Mathematics}, 21(4):514--532, 1971.

\bibitem{Huf52}
D.~Huffman.
\newblock A method for the construction of minimum-redundancy codes.
\newblock {\em Proceedings of the I.R.E.}, 40(9):1090--1101, 1952.

\bibitem{KP11}
J.~K{\"a}rkk{\"a}inen and S.~J. Puglisi.
\newblock Fixed block compression boosting in {FM}-indexes.
\newblock In {\em Proc. 18th International Symposium on String Processing and
  Information Retrieval (SPIRE)}, pages 174--184, 2011.

\bibitem{KN09}
M.~Karpinski and Y.~Nekrich.
\newblock A fast algorithm for adaptive prefix coding.
\newblock {\em Algorithmica}, 55(1):29--41, 2009.

\bibitem{KN76}
G.~O.~H. Katona and T.~O.~H. Nemetz.
\newblock {Huffman} codes and self-information.
\newblock {\em IEEE Transactions on Information Theory}, 22(3):337--340, 1976.

\bibitem{Knu73}
D.~E. Knuth.
\newblock {\em The Art of Computer Programming. Vol. 3: Sorting and Searching}.
\newblock Addison-Wesley, 1973.

\bibitem{LH90}
L.~L. Larmore and D.~S. Hirschberg.
\newblock A fast algorithm for optimal length-limited {H}uffman codes.
\newblock {\em Journal of the ACM}, 37(3):464--473, 1990.

\bibitem{LMSPE06}
M.~Liddell and A.~Moffat.
\newblock Decoding prefix codes.
\newblock {\em Software: Practice and Experience}, 36(15):1687--1710, 2006.

\bibitem{ML01}
R.~L. Milidi\'u and E.~S. Laber.
\newblock Bounding the inefficiency of length-restricted prefix codes.
\newblock {\em Algorithmica}, 31(4):513--529, 2001.

\bibitem{MREMD03}
R.~L. Milidi{\'u}, E.~S. Laber, L.~O. Moreno, and J.~C Duarte.
\newblock A fast decoding method for prefix codes.
\newblock In {\em Proc. 13th Data Compression Conference (DCC)}, page 438,
  2003.

\bibitem{Mof89}
A.~Moffat.
\newblock Word-based text compression.
\newblock {\em Software Practice and Experience}, 19(2):185--198, 1989.

\bibitem{MT97}
A.~Moffat and A.~Turpin.
\newblock On the implementation of minimum-redundancy prefix codes.
\newblock {\em IEEE Transactions on Communications}, 45(10):1200--1207, 1997.

\bibitem{MNZBY00}
E.~Moura, G.~Navarro, N.~Ziviani, and R.~{Baeza-Yates}.
\newblock Fast and flexible word searching on compressed text.
\newblock {\em ACM Transactions on Information Systems}, 18(2):113--139, 2000.

\bibitem{MR01}
J.~I. Munro and V.~Raman.
\newblock Succinct representation of balanced parentheses and static trees.
\newblock {\em SIAM Journal on Computing}, 31(3):762--776, 2001.

\bibitem{Nak91}
N.~Nakatsu.
\newblock Bounds on the redundancy of binary alphabetical codes.
\newblock {\em IEEE Transactions on Information Theory}, 37(4):1225--1229,
  1991.

\bibitem{Navjea08}
G.~Navarro.
\newblock Implementing the {LZ}-index: Theory versus practice.
\newblock {\em ACM Journal of Experimental Algorithmics (JEA)}, 13:article 2,
  2009.

\bibitem{Nav14}
G.~Navarro.
\newblock Wavelet trees for all.
\newblock {\em Journal of Discrete Algorithms}, 25:2--20, 2014.

\bibitem{NO13}
G.~Navarro and A.~Ord{\'o}{\~n}ez.
\newblock Compressing {H}uffman models on large alphabets.
\newblock In {\em Proc. 23rd Data Compression Conference (DCC)}, pages
  381--390, 2013.

\bibitem{NPsea12.1}
G.~Navarro and E.~Providel.
\newblock Fast, small, simple rank/select on bitmaps.
\newblock In {\em Proc. 11th International Symposium on Experimental Algorithms
  (SEA)}, LNCS 7276, pages 295--306, 2012.

\bibitem{PT08}
M.~Patrascu and M.~Thorup.
\newblock Time-space trade-offs for predecessor search.
\newblock {\em CoRR}, cs/0603043v1, 2008.
\newblock {\tt http://arxiv.org/pdf/cs/0603043v1}.

\bibitem{SK64}
E.~S. Schwarz and B.~Kallick.
\newblock Generating a canonical prefix encoding.
\newblock {\em Communications of the ACM}, 7(3):166--169, 1964.

\bibitem{She92}
D.~Sheinwald.
\newblock On binary alphabetic codes.
\newblock In {\em Proc. 2nd Data Compression Conference (DCC)}, pages 112--121,
  1992.

\bibitem{SIEMIIPL88}
A.~Siemi{\'n}ski.
\newblock Fast decoding of the {H}uffman codes.
\newblock {\em Information Processing Letters}, 26(5):237--241, 1988.

\bibitem{TM00}
A.~Turpin and A.~Moffat.
\newblock Housekeeping for prefix coding.
\newblock {\em IEEE Transactions on Communications}, 48(4):622--628, 2000.

\bibitem{vEBKZ77}
P.~van Emde~Boas, R.~Kaas, and E.~Zijlstra.
\newblock Design and implementation of an efficient priority queue.
\newblock {\em Mathematical Systems Theory}, 10:99--127, 1977.

\bibitem{WMB99}
I.~H. Witten, A.~Moffat, and T.~C. Bell.
\newblock {\em Managing Gigabytes: Compressing and Indexing Documents and
  Images}.
\newblock Morgan Kaufmann, 2nd edition, 1999.

\end{thebibliography}

\end{document}